\documentclass[11pt,a4paper]{article}

\usepackage{amssymb,amsthm,amsmath,mathtools}
\usepackage{thmtools,thm-restate}
\usepackage[round]{natbib}
\usepackage{pifont}

\usepackage[usenames,svgnames,xcdraw,table]{xcolor}
\definecolor{DarkGreen}{rgb}{0.1,0.5,0.1}
\usepackage[backref=page]{hyperref}
\hypersetup{
	colorlinks=true,
	linkcolor=red,
	urlcolor=DarkGreen,
	citecolor=blue
}
\renewcommand*{\backref}[1]{}
\renewcommand*{\backrefalt}[4]{%
    \ifcase #1 (Not cited.)%
    \or        (Cited on page~#2)%
    \else      (Cited on pages~#2)%
    \fi}
\usepackage[capitalise,noabbrev]{cleveref}
\Crefname{property}{Property}{Properties}
\Crefname{example}{Example}{Examples}
\Crefname{table}{Table}{Tables}

\usepackage{cancel}
\usepackage{nicefrac}
\usepackage{enumerate}
\usepackage{etoolbox}
\usepackage[margin=0.99in]{geometry}
\usepackage[T1]{fontenc}
\usepackage[tt=false]{libertine}
\usepackage[spacing=true,factor=1100,stretch=10,shrink=10]{microtype}
\microtypecontext{spacing=nonfrench}

\usepackage{url}
\usepackage[utf8]{inputenc}
\usepackage[small]{caption}
\usepackage{graphicx}
\usepackage{booktabs}
\usepackage{mathrsfs,amsfonts,dsfont,authblk}
\usepackage{array,multirow,graphicx,bigdelim}

\usepackage[linesnumbered,lined,boxed,ruled,vlined]{algorithm2e}

\SetKwInOut{Parameters}{Parameters}
\SetKwComment{Comment}{$\triangleright$\ }{}
\SetAlFnt{\small}
\SetAlCapFnt{\small}
\SetAlCapNameFnt{\small}
\SetAlCapHSkip{0pt}
\IncMargin{-\parindent}

\usepackage{tikz}
\usepackage{tikz-cd}
\usetikzlibrary{fit,calc,shapes,decorations.text,arrows,decorations.markings,decorations.pathmorphing,shapes.geometric,positioning,decorations.pathreplacing}
\tikzset{snake it/.style={decorate, decoration=snake}}

\newcommand*{\tikzmk}[1]{\tikz[remember picture,overlay,] \node (#1) {};\ignorespaces}
\newcommand{\boxit}[1]{\tikz[remember picture,overlay]{\node[yshift=3pt,xshift=4pt,fill=#1,opacity=.25,fit={(A)($(B)+(1.0\linewidth,.8\baselineskip)$)}] {};}\ignorespaces}
\colorlet{mygray}{gray!40}

\makeatletter
\let\oldnl\nl
\newcommand{\nonl}{\renewcommand{\nl}{\let\nl\oldnl}}
\makeatother

  {\list{}{\leftmargin=#1\rightmargin=#1}\item[]}%
  {\endlist}

\usepackage[labelfont={normalfont,bf},textfont=it]{caption}
\usepackage{subcaption}

\newtheorem{lemma}{Lemma}
\newtheorem{corollary}{Corollary}
\theoremstyle{definition}

\theoremstyle{remark}
\newtheorem{remark}{Remark}

\usepackage{flushend}
\usepackage[utf8]{inputenc}
\usepackage[inline]{enumitem}
\Crefname{claim}{Claim}{Claims}

\allowdisplaybreaks
\renewcommand{\O}{\mathcal{O}}
\usepackage[compact]{titlesec}

\let\displaystyle\textstyle

\newcommand\floor[1]{\lfloor#1\rfloor}
\newcommand\ceil[1]{\lceil#1\rceil}

\PassOptionsToPackage{hyphens}{url}
\usepackage{etoolbox}
\usepackage[title]{appendix} 

\usepackage[linesnumbered,lined,boxed,ruled]{algorithm2e}
\SetKwInOut{Parameter}{Parameter}
\SetKwComment{Comment}{$\triangleright$\ }{}
\SetKwRepeat{Do}{do}{while}
\SetKwFor{ForEach}{for each}{do}{endfor}
\SetCommentSty{mycommfont}

\usepackage{nicefrac}
\usepackage{enumerate}

\usepackage{subcaption}

\usepackage{tikz}
\usetikzlibrary{math,shapes,decorations,arrows,positioning,fit,calc}

\usepackage[T1]{fontenc}
\usepackage[utf8]{inputenc}

\allowdisplaybreaks

\newcommand{\C}{{\mathcal{C}}}
\newcommand{\CC}{\ensuremath{\mathcal C}\xspace}

\newcommand{\EF}[1]{\ifstrempty{#1}{\textup{EF}}{\textup{EF{#1}}}}

\newcommand{\EFX}{\textnormal{\textup{EFX}}}
\newcommand{\EQ}[1]{\ifstrempty{#1}{\textup{EQ}}{\textup{EQ{#1}}}}

\newcommand{\I}{\mathcal{I}}

\newcommand{\N}{{\mathcal{N}}}

\newcommand{\NPC}{\textnormal{\textup{NP-complete}}}
\newcommand{\NPH}{\textnormal{\textup{NP-hard}}}
\newcommand{\NPHness}{\textnormal{\textup{NP-hardness}}}
\newcommand{\NW}{\textup{NW}}

\renewcommand{\O}{{\mathcal{O}}}
\newcommand{\OO}{\ensuremath{\mathcal O}\xspace}

\newcommand{\PO}{\textup{PO}}

\newcommand{\U}{\mathcal{U}}

\let\displaystyle\textstyle

\newcommand{\LES}{\textup{\textsc{LNES}}}
\newcommand{\LESfull}{\textup{\textsc{Linear Near-Exact Satisfiability}}}

\newcommand{\true}{\textsc{true}\xspace}

\newcommand{\pr}{\ensuremath{\prime}\xspace}

\usepackage[backgroundcolor=Moccasin,colorinlistoftodos]{todonotes}
\newcounter{nmcomment}

\newcounter{rvcomment}

\newcounter{cscomment}

\title{Equitable Division of a Path}

\author[1]{Neeldhara Misra}
\author[2]{Chinmay Sonar}
\author[3]{P. R. Vaidyanathan}
\author[4]{Rohit Vaish}
\affil[1]{Indian Institute of Technology Gandhinagar, India\\
	{\small\texttt{neeldhara.m@iitgn.ac.in}}}
\affil[2]{University of California, Santa Barbara, USA\\
	{\small\texttt{chinmaysonar96@gmail.com}}}
\affil[3]{Technische Universit\"{a}t Wien, Austria\\
	{\small\texttt{vaidyanathan@ac.tuwien.ac.at}}}
\affil[4]{Tata Institute of Fundamental Research\\
	{\small\texttt{rohit.vaish@tifr.res.in}}}

\date{}

\begin{document}

\maketitle

\begin{abstract}
We study fair resource allocation under a connectedness constraint wherein a set of indivisible items are arranged on a \emph{path} and only connected subsets of items may be allocated to the agents. An allocation is deemed fair if it satisfies \emph{equitability up to one good} (\EQ{1}), which requires that agents' utilities are approximately equal. We show that achieving \EQ{1} in conjunction with well-studied measures of \emph{economic efficiency} (such as Pareto optimality, non-wastefulness, maximum egalitarian or utilitarian welfare) is computationally hard even for \emph{binary} additive valuations. On the algorithmic side, we show that by relaxing the efficiency requirement, a connected \EQ{1} allocation can be computed in polynomial time for \emph{any} given ordering of agents, even for general monotone valuations. Interestingly, the allocation computed by our algorithm has the highest egalitarian welfare among all allocations consistent with the given ordering. On the other hand, if efficiency is required, then tractability can still be achieved for binary additive valuations with \emph{interval structure}. On our way, we strengthen some of the existing results in the literature for other fairness notions such as envy-freeness up to one good (\EF{1}), and also provide novel results for negatively-valued items or \emph{chores}.
\end{abstract}

\section{Introduction}
\label{sec:Introduction}

The question of how to \emph{fairly} divide a set of resources among agents has been extensively studied in economics, mathematics, and computer science. The formal treatment of such resource allocation problems---commonly referred to as \emph{fair division}---dates back several decades~\citep{S48problem}. There is now a rich literature on fair division problems~\citep{BT96fair,M04fair,BCE+16handbook}, comprising of a variety of solution concepts and associated existential and computational results. Many of these insights have found impressive practical applications such as in rent division~\citep{S99rental}, credit assignment~\citep{dCMT08impartial}, and cluster computing~\citep{GZH+11dominant}.

Many real-world resource allocation problems exhibit a natural spatial or temporal structure, and in such scenarios, it is desirable to have \emph{contiguous} allocations. For example, when allocating supercomputing time, a contiguous processing window is preferable over one that involves multiple restarts. Similarly, when assigning office space in a department building, each research group might prefer a contiguous segment of rooms for ease of communication.

\begin{sloppypar}
In this work, we study the seemingly conflicting goals of fairness and contiguity in the context of allocating \emph{indivisible} resources (or goods). Specifically, we consider a set of indivisible goods that are represented by the vertices of a \emph{path} graph, and require that each agent is allocated a connected subgraph. Fair allocation of indivisible goods has received growing interest within both artificial intelligence as well as theoretical computer science literature~\citep{AGM+15fair,BL16characterizing,CG18approximating,KPW18fair,CKM+19unreasonable},
 motivated, in part, by notable real-world applications such as course allocation~\citep{B11combinatorial} and property division ~\citep{PW12divorcing}. The research area has been further popularized by the website \emph{Spliddit} ({\small{\url{http://www.spliddit.org/}}}) that provides implementations of provably fair algorithms for a wide array of resource allocation problems~\citep{GP15spliddit}.
\end{sloppypar}

While there are countless formulations of what it means to be fair, each with its own merit, in this work we focus on one well-established notion of fairness called \emph{equitability}~\citep{DS61cut}. An equitable allocation is one in which agents derive equal utilities from their assigned shares. 
Equitability is a particularly compelling fairness criterion in settings such as dividing climate change responsibilities among countries~\citep{T02fair} and in designing taxation policies. It also enjoys empirical support, as lab experiments and an online user study have found that equitability---or ``aversion of interpersonal inequity'''---can be an important predictor of the perceived fairness of an allocation, possibly more so than the classic ``intrapersonal'' criterion of envy-freeness~\citep{HP09envy,GMP+17fairest}.
Equitability is also a key property in the well-known \emph{adjusted winner} algorithm~\citep{BT96fair} which has been applied to divorce settlements.

For indivisible items, perfect equitability may not be possible, which motivates the need for a natural relaxation called \emph{equitability up to one good} (\EQ{1})~\citep{FSV+19equitable}. This notion requires that the inequity between any pair of agents can be eliminated by removing some item from the happier agent's bundle. Since an empty allocation is vacuously fair, the study of fairness notions is often coupled with \emph{economic efficiency}. To this end, we study \EQ{1} alongside various efficiency measures such as Pareto optimality, non-wastefulness, and maximum egalitarian (max-min) or utilitarian (sum) welfare (see Preliminaries for the relevant definitions).

The study of connected fair allocations of general graphs was initiated by \citet{BCE+17fair} with a focus on other fairness notions such as envy-freeness, proportionality, and maximin share. Concurrently, \citet{S19fairly} showed that for a path graph, a connected and approximately equitable allocation always exists and can be efficiently computed. This work also provided a non-constructive proof of existence of egalitarian-optimal and approximately equitable allocations, but did not consider other efficiency notions. Importantly, the notion of approximate equitability in Suksompong's work is strictly weaker than \EQ{1}, and as we observe later, 
his algorithm could fail to find \EQ{1} allocations even when such allocations are known to exist. Thus, the existential and computational questions pertaining to \EQ{1} allocations remain unanswered by prior work.

\textbf{Our Contributions:}
We initiate the study of \EQ{1} allocations under connectedness constraints 
and make the following contributions:
\begin{enumerate}
    \item \textbf{Hardness results for \EQ{1} and efficient allocations}: We show that checking the existence of a connected \EQ{1} allocation satisfying any of the aforementioned efficiency measures is \NPH{} even under \emph{binary} additive valuations (\Cref{thm:NW_NPcomplete,cor:EQ1_NW_Egal_Util,thm:EF1_EQ1_PO}). 
    All of our results follow from a \emph{single} construction that also has implications for other fairness notions such as \emph{envy-freeness up to one good} (\EF{1}) as well as negatively-valued items (or \emph{chores}).
    
    \item \textbf{Algorithmic result for complete \EQ{1} allocations}: By relaxing the efficiency condition and only requiring \emph{completeness} (i.e., not leaving any good unassigned), we obtain a polynomial-time algorithm for computing a connected \EQ{1} allocation whose egalitarian welfare is the highest among all allocations that are consistent with a given ordering of agents (\Cref{thm:EQ1+complete-polytime}). This resolves an open problem of \citet{S19fairly}. Notably, our algorithm applies to any instance with monotone (possibly non-additive) valuations.
    
    \item \textbf{Structured preferences}: We provide an efficient algorithm for checking the existence of a connected, non-wasteful, and \EQ{1} allocation when agents have binary additive valuations with \emph{extremal} interval structure~(\Cref{thm:EQoneNW_Extremal}).
\end{enumerate}
\section{Related Work}
\label{sec:Related_Work}

Fair division problems have been classically studied in the context of \emph{divisible} resources, most prominently in the \emph{cake-cutting} literature; see \citep[Chapter 13]{BCE+16handbook} for an excellent survey. There is also a vast literature on connected (or contiguous) cake-cutting, spanning various notions of fairness and economic efficiency~\citep{S80cut,S99rental,DQS12algorithmic,BCH+12optimal,ADH13computing,AD15efficiency,SS18resource,BN19communication,GHS20contiguous}. In particular, for equitability, it is known that for \emph{any} given ordering of the agents, there exists a connected equitable division of a cake consistent with the ordering~\citep{CDP13existence}. Although no finite procedure can compute an exactly equitable division even without the connectedness constraint~\citep{PW17lower}, 
it is known that an $\varepsilon$-equitable connected division can be computed using finite protocols~\citep{CP12computability}. Equitability has also been studied in combination with other fairness notions. For example, while there always exists a connected equitable division that is also \emph{proportional}~\citep{CDP13existence}, there might not exist a connected division that is simultaneously equitable and \emph{envy-free}~\citep{BJK06better}.

\begin{sloppypar}
For \emph{indivisible} resources, the study of connected fair division has more recent origins~\citep{MT14envy,BCE+17fair,S19fairly}. A number of fairness notions such as proportionality, envy-freeness, and maximin share have been examined in this model when the resources are \emph{goods}~\citep{BCE+17fair,LT18maximin,IP19Pareto,BCF+19almost,BIL+21connected,S19fairly,OPS19fairly}, \emph{chores}~\citep{BCL19chore}, and \emph{mixed} items involving both goods and chores~\citep{ACI+19fair}. A noteworthy result in this context concerns the existence of allocations satisfying \emph{envy-freeness up to one good} (\EF{1}) when the number of agents is at most four~\citep{BCF+19almost}, or when agents have identical valuations~\citep{BCF+19almost,OPS19fairly}. As we observe in \Cref{rem:EF1_identical}, the latter result follows as a corollary of our main algorithmic result.
\end{sloppypar}

For indivisible goods without the connectedness requirement, \citet{GMT14near} provided an efficient algorithm for achieving \emph{equitability up to any good}. Subsequently, \citet{FSV+19equitable} studied (approximate) equitability along with Pareto optimality. Among other results, they showed that an \EQ{1} and Pareto optimal allocation might fail to exist even with binary valuations, and provided efficient algorithms for checking the existence of such allocations. By contrast, as we show in \Cref{thm:EF1_EQ1_PO}, the problem becomes \NPC{} when connectedness is also required.

\section{Preliminaries}
\label{sec:Preliminaries}
Let $\N = \{a_1,a_2,\dots,a_n\}$ be a set of $n \in \mathbb{N}$ \emph{agents}, and $G = (V,E)$ be an undirected graph. Each vertex $v \in V$ of the graph $G$ corresponds to an \emph{indivisible good} (or \emph{item}) with $m \coloneqq |V|$ goods overall. A \emph{(connected) bundle} is a set of goods $S \subseteq V$ whose corresponding vertices induce a connected subgraph of $G$. We let $\C(V) \subseteq 2^V$ denote the set of all connected subsets of $V$. 
Unless stated otherwise, we will assume that $G$ is a \emph{path} given by $\{v_1,v_2,\dots,v_m\}$ where $\{v_i,v_{i+1}\} \in E$ for $i \in [m-1]$.

A (connected) \emph{allocation} $A : \N \rightarrow \C(V)$ assigns to each agent $a_i$ a connected bundle $A(a_i) \in \C(V)$ such that no good is assigned to more than one agent. We will denote an allocation as an ordered tuple $A = (A_1,A_2,\dots,A_n)$, where $A_i \coloneqq A(a_i)$. An allocation is said to be \emph{complete} if it does not leave any good unassigned; that is, for any good $v$, there exists some agent $a_i$ such that $v \in A_i$. A \emph{partial} allocation is one that is not complete. Unless stated explicitly otherwise, the term `allocation' will refer to a complete allocation.

The preferences of agent $a_i$ are specified by a \emph{valuation function} $u_i : \C(V) \rightarrow \mathbb{N} \cup \{0\}$.
We say that the valuation functions are \emph{monotone} if for any pair of connected bundles $S,S' \in \C(V)$ such that $S \subseteq S'$, we have $u_i(S) \leq u_i(S')$. The valuation functions are said to be \emph{additive} if for each agent $a_i$ and each bundle $S \in \C(V)$, $u_i(S) \coloneqq \sum_{v \in S} u_i(\{v\})$, where $u_i(\emptyset) \coloneqq 0$. Note that since all valuations are non-negative, any additive valuation function is also monotone. We will assume throughout that the valuations are additive (however, note that our algorithmic results apply to monotone, possibly non-additive valuations). For simplicity, we will write $u_{i,j} \coloneqq u_i(\{v_j\})$. 
An $n$-tuple of valuation functions $\U = \{u_1,\dots,u_n\}$ is called a \emph{valuation profile}. We say that agents have \emph{binary} (additive) valuations if $u_{i,j} \in \{0,1\}$ for all $a_i \in \N$ and $v_j \in V$.

\textbf{Fairness notions:} An allocation $A$ is said to be
\begin{itemize}
    \item \emph{equitable} $(\EQ{})$ if for every pair of agents $a_i, a_k \in \N$, the utilities of $a_i$ and $a_k$ for their respective bundles are equal, that is, $u_i(A_i) = u_k(A_k)$,
    \item \emph{equitable up to one good} $(\EQ{1})$ if for every pair of agents $a_i, a_k \in \N$ such that $A_k \neq \emptyset$, there exists some good $v \in A_k$ such that $u_i(A_i) \geq u_k(A_k \setminus \{v\})$,
    \item \emph{envy-free} $(\EF)$ if for every pair of agents $a_i, a_k \in \N$, $u_i(A_i) \geq u_i(A_k)$, and
    \item \emph{envy-free up to one good} $(\EF{1})$ if for every pair of agents $a_i, a_k \in \N$, $u_i(A_i) \geq u_i(A_k \setminus \{v\})$ for some $v \in A_k$.
\end{itemize}
\begin{sloppypar}
\noindent The notions of \EQ{}, \EQ{1}. \EF{}, and \EF{1} were formulated in the context of resource allocation by \citet{DS61cut}, \citet{FSV+19equitable}, \citet{F67resource}, and \citet{B11combinatorial}, respectively.\footnote{\citet{LMM+04approximately} studied a weaker approximation of envy-freeness than \EF{1}, but their algorithm is known to compute an \EF{1} allocation.}
\end{sloppypar}

Notice that equitability and envy-freeness (and their corresponding relaxations) coincide when agents have \emph{identical} valuations (i.e., if $u_i=u_k$ for every $a_i,a_k \in \N$) but are incomparable in general. Although our focus in this paper is on (approximate) equitability, some of our results also have implications for (approximate) envy-freeness.

\textbf{Efficiency notions:}
An allocation $A$ is said to be
\begin{itemize}
    \item \emph{Pareto optimal} $(\PO{})$ if for no other connected allocation $B$, we have $u_i(B_i) \geq u_i(A_i)$ for every agent $a_i$, with at least one of the inequalities being strict, and 
    \item \emph{non-wasteful} $(\NW{})$ if for any good $v$, there exists some agent $a_i$ such that $v \in A_i$ and $u_i(\{v\}) > 0$.\footnote{To make this notion well-defined, we will assume throughout that in any given instance, for every good there is at least one agent with a non-zero value for it. This assumption is without loss of generality as our negative results (pertaining to computational hardness and non-existence) hold even under this assumption, and our positive results (algorithms and existence results) do not need this assumption.}
\end{itemize}
The \emph{utilitarian welfare} of $A$ is the sum of utilities of all agents in $A$, i.e., $\sum_{a_i \in \N} u_i(A_i)$, and the \emph{egalitarian welfare} of $A$ is the utility of the least happy agent, i.e., $\min_{a_i \in \N} u_i(A_i)$.

Non-wastefulness and Pareto optimality are, in general, incomparable notions even when $G$ is a path.\footnote{Consider three goods $v_1,v_2,v_3$ on a path and two agents with valuations $u_1 = (1,10,0)$ and $u_2 = (10,1,1)$. The allocation $A \coloneqq (\{v_1\},\{v_2,v_3\})$ is non-wasteful but is Pareto dominated by the (wasteful) allocation $B \coloneqq (\{v_2,v_3\},\{v_1\})$.}
However, for binary valuations, $\NW{} \Rightarrow \PO{} \Rightarrow$ complete (since, for binary valuations, a non-wasteful allocation maximizes the utilitarian social welfare and is therefore Pareto optimal), and there are simple examples where these implications are strict.

\textbf{Connected fair division problem:} The input to this problem is a tuple $\I = \langle G, \N, \U \rangle$ consisting of a graph $G$, a set of agents $\N$, and a valuation profile $\U$. The goal is to determine whether $\I$ admits a \emph{connected} allocation satisfying the desired notions of fairness and efficiency. Notice that if $G$ is a clique, we recover the standard fair division model without the connectedness constraint. In this work, we will exclusively focus on the case where $G$ is a path graph.

\emph{$(a,b)$-sparse instances:} Given any $1 \leq a \leq m$ and $1 \leq b \leq n$, we say that an instance with binary valuations is \emph{$(a,b)$-sparse} if each agent approves at most $a$ goods and each good is approved by at most $b$ agents.

\section{Hardness Results for \EQ{1} and Efficient Allocations}
\label{sec:Hardness_results}

Note that in the absence of the connectedness constraint, a non-wasteful allocation can be easily computed by assigning each good to an agent that has a positive value for it. 
By contrast, connectedness poses a substantial computational challenge even when we are only looking to satisfy non-wastefulness (without any fairness constraints), as the problem turns out to be \NPC{} (\Cref{thm:NW_NPcomplete}).

\begin{figure*}
    \centering
    \begin{center}
    \begin{tikzpicture}
    \tikzmath{\x=0;\edgex=0;\edgewidth=0.18;\nodesep=0.65;\edgenodesep=0.23;\yoffset=0;\xoffset=0;}
    \draw (\x,0) circle [radius=7pt] node [] {\scriptsize{$U^{}_1$}};
    \tikzmath{\edgex=\x+\edgenodesep;}
    \tikzmath{\x=\x+\nodesep;}
    \draw (\edgex,0) -- (\edgex+\edgewidth,0);
    \draw (\x,0) circle [radius=7pt] node [] {\scriptsize{$V^{}_1$}};
    \tikzmath{\edgex=\x+\edgenodesep;}
    \tikzmath{\x=\x+\nodesep;}
    \draw (\edgex,0) -- (\edgex+\edgewidth,0);
    \draw (\x,0) circle [radius=7pt] node [] {\scriptsize{$U_1^\pr$}};
    \tikzmath{\edgex=\x+\edgenodesep;}
    \tikzmath{\x=\x+\nodesep;}
    \draw (\edgex,0) -- (\edgex+\edgewidth,0);
    \draw (\x,0) circle [radius=7pt] node [] {\scriptsize{$V_1^\pr$}};
    \tikzmath{\edgex=\x+\edgenodesep;}
    \tikzmath{\x=\x+\nodesep;}
    \draw (\edgex,0) -- (\edgex+\edgewidth,0);
    \draw (\x,0) node [] {\scriptsize{$\dots$}};
    \tikzmath{\edgex=\x+\edgenodesep;}
    \tikzmath{\x=\x+\nodesep;}
    \draw (\edgex,0) -- (\edgex+\edgewidth,0);
    \draw (\x,0) circle [radius=7pt] node [] {\scriptsize{$U^{}_p$}};
    \tikzmath{\edgex=\x+\edgenodesep;}
    \tikzmath{\x=\x+\nodesep;}
    \draw (\edgex,0) -- (\edgex+\edgewidth,0);
    \draw (\x,0) circle [radius=7pt] node [] {\scriptsize{$V^{}_p$}};
    \tikzmath{\edgex=\x+\edgenodesep;}
    \tikzmath{\x=\x+\nodesep;}
    \draw (\edgex,0) -- (\edgex+\edgewidth,0);
    \draw (\x,0) circle [radius=7pt] node [] {\scriptsize{$U_p^\pr$}};
    \tikzmath{\edgex=\x+\edgenodesep;}
    \tikzmath{\x=\x+\nodesep;}
    \draw (\edgex,0) -- (\edgex+\edgewidth,0);
    \draw (\x,0) circle [radius=7pt] node [] {\scriptsize{$V_p^\pr$}};
    \tikzmath{\edgex=\x+\edgenodesep;}
    \tikzmath{\x=\x+\nodesep;}
    \draw (\edgex,0) -- (\edgex+\edgewidth,0);
    \draw[fill=black!10] (\x,0) circle [radius=7pt] node [] {\scriptsize{$S^{}_0$}};
    \tikzmath{\edgex=\x+\edgenodesep;}
    \tikzmath{\x=\x+\nodesep;}
    \draw (\edgex,0) -- (\edgex+\edgewidth,0);
    \draw (\x,0) circle [radius=7pt] node [] {\scriptsize{$C_1^L$}};
    \tikzmath{\edgex=\x+\edgenodesep;}
    \tikzmath{\x=\x+\nodesep;}
    \draw (\edgex,0) -- (\edgex+\edgewidth,0);
    \draw[fill=black!10] (\x,0) circle [radius=7pt] node [] {\scriptsize{$S^{}_1$}};
    \tikzmath{\edgex=\x+\edgenodesep;}
    \tikzmath{\x=\x+\nodesep;}
    \draw (\edgex,0) -- (\edgex+\edgewidth,0);
    \draw (\x,0) circle [radius=7pt] node [] {\scriptsize{$C_1^R$}};
    \tikzmath{\edgex=\x+\edgenodesep;}
    \tikzmath{\x=\x+\nodesep;}
    \draw (\edgex,0) -- (\edgex+\edgewidth,0);
    \draw (\x,0) circle [radius=7pt] node [] {\scriptsize{$C_2^L$}};
    \tikzmath{\edgex=\x+\edgenodesep;}
    \tikzmath{\x=\x+\nodesep;}
    \draw (\edgex,0) -- (\edgex+\edgewidth,0);
    \draw[fill=black!10] (\x,0) circle [radius=7pt] node [] {\scriptsize{$S^{}_2$}};
    \tikzmath{\edgex=\x+\edgenodesep;}
    \tikzmath{\x=\x+\nodesep;}
    \draw (\edgex,0) -- (\edgex+\edgewidth,0);
    \draw (\x,0) circle [radius=7pt] node [] {\scriptsize{$C_2^R$}};
    \tikzmath{\edgex=\x+\edgenodesep;}
    \tikzmath{\x=\x+\nodesep;}
    \draw (\edgex,0) -- (\edgex+\edgewidth,0);
    \draw (\x,0) node [] {\scriptsize{$\dots$}};
	%
    %
    \tikzmath{\edgex=\x+\edgenodesep-\xoffset;}
    \tikzmath{\x=\x+\nodesep-\xoffset;}
    \draw (\edgex,\yoffset) -- (\edgex+\edgewidth,\yoffset);
    \draw (\x,\yoffset) circle [radius=7pt] node [] {\scriptsize{$C_p^L$}};
    \tikzmath{\edgex=\x+\edgenodesep;}
    \tikzmath{\x=\x+\nodesep;}
    \draw (\edgex,\yoffset) -- (\edgex+\edgewidth,\yoffset);
    \draw[fill=black!10] (\x,\yoffset) circle [radius=7pt] node [] {\scriptsize{$S_p$}};
    \tikzmath{\edgex=\x+\edgenodesep;}
    \tikzmath{\x=\x+\nodesep;}
    \draw (\edgex,\yoffset) -- (\edgex+\edgewidth,\yoffset);
    \draw (\x,\yoffset) circle [radius=7pt] node [] {\scriptsize{$C_p^R$}};
    \tikzmath{\edgex=\x+\edgenodesep;}
    \tikzmath{\x=\x+\nodesep;}
    \draw (\edgex,\yoffset) -- (\edgex+\edgewidth,\yoffset);
    \draw (\x,\yoffset) circle [radius=7pt] node [] {\scriptsize{$D_1^{}$}};
    \tikzmath{\edgex=\x+\edgenodesep;}
    \tikzmath{\x=\x+\nodesep;}
    \draw (\edgex,\yoffset) -- (\edgex+\edgewidth,\yoffset);
    \draw (\x,\yoffset) circle [radius=7pt] node [] {\scriptsize{$D_1^\pr$}};
    \tikzmath{\edgex=\x+\edgenodesep;}
    \tikzmath{\x=\x+\nodesep;}
    \draw (\edgex,\yoffset) -- (\edgex+\edgewidth,\yoffset);
    \draw (\x,\yoffset) node [] {\scriptsize{$\dots$}};
    \tikzmath{\edgex=\x+\edgenodesep;}
    \tikzmath{\x=\x+\nodesep;}
    \draw (\edgex,\yoffset) -- (\edgex+\edgewidth,\yoffset);
    \draw (\x,\yoffset) circle [radius=7pt] node [] {\scriptsize{$D_p^{}$}};
    \tikzmath{\edgex=\x+\edgenodesep;}
    \tikzmath{\x=\x+\nodesep;}
    \draw (\edgex,\yoffset) -- (\edgex+\edgewidth,\yoffset);
    \draw (\x,\yoffset) circle [radius=7pt] node [] {\scriptsize{$D_p^\pr$}};
    \end{tikzpicture}
    \end{center}
\caption{The instance used in the proof of \Cref{thm:NW_NPcomplete}.
}
\label{fig:NW_goods_order}
\end{figure*}
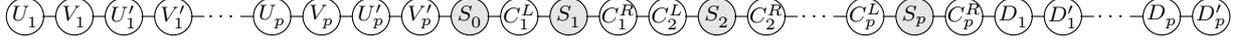

\begin{restatable}[]{theorem}{NWNPcomplete}
 \label{thm:NW_NPcomplete}

 Determining whether there exists a connected non-wasteful allocation is \NPC{} for a path and a $(4,4)$-sparse binary valuations instance.
\end{restatable}

To prove \Cref{thm:NW_NPcomplete}, we will show a reduction from a structured version of \textsc{Satisfiability} called \LESfull{} (\LES{}) which is known to be \NPC{}~\citep{DM19deleting}. An instance of \LES{} consists of $5p$ clauses (where $p \in \mathbb{N}$) denoted as follows:
$$ \CC = \{U_1, V_1, U_1^\pr, V_1^\pr, \cdots, U_p, V_p, U_p^\pr, V_p^\pr\} \cup \{ C_1, \cdots, C_p \}.$$
We will refer to the first $4p$ clauses as the \emph{core} clauses, and the remaining clauses as the \emph{auxiliary} clauses. The set of variables consists of $p$ \emph{main variables} $x_1,\dots,x_p$ and $4p$ \emph{shadow variables} $y_1,\dots,y_{4p}$. 

Each core clause consists of two literals and has the following structure:
$$\forall \, i \in [p], U_i^{} \cap V_i^{} = \{x_i\} \mbox{ and } U_i^\pr \cap V_i^\pr = \{\bar{x}_i\}.$$
Each main variable $x_i$ occurs exactly twice as a positive literal and exactly twice as a negative literal. The main variables only occur in the core clauses. Each shadow variable makes two appearances: as a positive literal in an auxiliary clause and as a negative literal in a core clause. Each auxiliary clause consists of four literals, each corresponding to a positive occurrence of a shadow variable. 

The \LES{} problem asks whether, given a set of clauses with the aforementioned structure, there exists an assignment $\tau$ of truth values to the variables such that \emph{exactly one} literal in every core clause and \emph{exactly two} literals in every auxiliary clause evaluate to \true{} under $\tau$.

\begin{proof} (of \Cref{thm:NW_NPcomplete}) Let $\phi$ be an instance of \LES{}. We will begin with a description of the reduced instance.

\textbf{Goods}: Introduce one good for every core clause, denoted by $U_i^{}$, $V_i^{}$, $U_i^\pr$, and $V_i^\pr$, and two goods for every auxiliary clause, denoted by $C_i^L$ and $C_i^R$. We refer to these as \emph{core} and \emph{auxiliary} goods, respectively. We also introduce $2p$ \emph{dummy} goods 
 $D_1^{},D_1^\pr, \ldots, D_p^{}, D_p^\pr$ as well as $p+1$ \emph{separator} goods $S_0, S_1, \ldots, S_p$. Thus, the total number of goods is $m = 4p+2p+2p+(p+1) = 9p+1$. The goods are arranged as shown in \Cref{fig:NW_goods_order}.

\textbf{Agents}: For every main variable $x_i$, we will introduce two agents $a_{x_i}$ and $a_{\bar{x}_i}$ for the two literals; these are referred to as \emph{main agents} of the \emph{positive} and \emph{negative} type, respectively. For every $i \in [p]$, the agent $a_{x_i}$ approves (i.e., values at $1$) the goods $U_i^{}, V_i^{}, D_i^{}, D_i^\pr$, while the agent $a_{\bar{x_i}}$ approves the goods $U_i^\pr, V_i^\pr, D_i^{}, D_i^\pr$. We also introduce a \emph{shadow agent} for every shadow variable. If $y$ is a shadow variable occurring in the core clause $U_i^{}$ and auxiliary clause $C_j^{}$, then the shadow agent corresponding to $y$ approves the goods $U_i^{}$, $C_j^L$, and $C_j^R$. The set of goods approved by $y$ is analogously defined if it appears in the core clauses $V_i^{}$, $U_i^\pr$, or $V_i^\pr$. Finally, we introduce $p+1$ \emph{separator} agents $s_0, s_1, \ldots, s_p$ such that for every $i \in \{0,1,\dots,p\}$, $s_i$ only approves the separator good $S_i$. Thus, the total number of agents is $n = 2p+4p+(p+1) = 7p+1$. Observe that the constructed instance is $(4,4)$-sparse. 
We now turn to the proof of equivalence of the two instances.

\textbf{The Forward Direction.} Let $\tau$ be a satisfying assignment for the \LES{} instance. We will construct the desired allocation as follows: For every $i \in [p]$, if the main variable $x_i$ evaluates to \true{} (i.e., if $\tau(x_{i}) = 1$), then assign $U_{i}^{}$ and $V_{i}^{}$ to agent $a_{x_{i}}$, $D_{i}^{}$ and $D_{i}^\pr$ to agent $a_{\bar{x}_{i}}$, and $U_{i}^\pr$ and $V_{i}^\pr$ to the (unique) shadow agents that approve these goods. Otherwise, if $\tau(x_i) = 0$, then assign $U_{i}^\pr$ and $V_{i}^\pr$ to agent $a_{\bar{x}_{i}}$, $D_{i}^{}$ and $D_{i}^\pr$ to agent $a_{x_{i}}$, and $U_{i}^{}$ and $V_{i}^{}$ to the (unique) shadow agents that approve these goods. Additionally, for every $i \in \{0,\dots,p\}$, assign $S_{i}$ to the agent $s_i$. Finally, for every $i \in [p]$, assign the goods $C_{i}^{L}$ and $C_{i}^{R}$ to the two shadow agents whose corresponding literals satisfy 
the auxiliary clause $C_{i}$.

The above allocation assigns each good to an agent that approves it and is therefore non-wasteful. It is also easy to see that the allocation is connected: The only agents that receive more than one good under this allocation are the main agents, and they always receive either two adjacent core goods or two adjacent dummy goods.

\textbf{The Reverse Direction.} We will now show how to  
recover an \LES{} assignment given a connected and non-wasteful allocation $A$ for the fair division instance.

Observe that due to non-wastefulness, each separator good is assigned to a unique separator agent, and 
the separator agents are not assigned any other goods.
Thus, for every $i \in \{0,1,\dots,p\}$, $A_{s_i} = \{S_i\}$. Similarly, the $2p$ dummy goods $D_1^{},D_1^\pr, \dots,D_p^{},D_p^\pr$ must be allocated among \emph{at least} $p$ main agents, which leaves at most $p$ main agents for receiving the core goods. Furthermore, the separator goods prevent any shadow agent from getting more than one auxiliary good. Thus, the $2p$ auxiliary goods are assigned to \emph{exactly} $2p$ shadow agents, leaving the other $2p$ shadow agents for receiving the core goods.

Since each core good is approved by a unique shadow agent, at most $2p$ core goods can be allocated among shadow agents. Thus, the remaining $2p$ (or more) core goods should go to the main agents. However, due to non-wastefulness, a main agent cannot get more than two core goods. Overall, this means that one set of $p$ main agents gets exactly two core goods each (the ``lucky'' agents), while the other set of $p$ main agents gets two dummy goods each (the ``unlucky'' agents). Notice that the two main agents corresponding to a main variable cannot both be lucky (since that would leave one or more dummy goods unassigned), nor can both be unlucky (as that would create a similar violation for the core goods).

This brings us to a natural way of deriving an assignment $\tau$ from the allocation $A$. If the main agent of the positive (respectively, negative) type is unlucky, then we let $\tau(x_i) = 0$ (respectively, $\tau(x_i) = 1$). Furthermore, if $A$ allocates a core good to a shadow agent, then the corresponding shadow variable is set to $0$, while shadow variables corresponding to shadow agents who receive auxiliary goods are set to $1$. Note that exactly $2p$ of the $4p$ shadow variables are set to $1$.  
It can be verified that $\tau$ is indeed a satisfying assignment.

\end{proof}
Notice that the allocation obtained in the forward direction in the proof of Theorem~\ref{thm:NW_NPcomplete} is \EQ{1} and \EF{1}, and the argument for the reverse direction is driven only by non-wastefulness. Thus, we also obtain hardness results for \EQ{1}+\NW{} and \EF{1}+\NW{} allocations. Additionally, for binary additive valuations, an allocation is non-wasteful if and only if its utilitarian welfare is at least~$m$. These observations establish the hardness of a number of related problems.
\begin{corollary}
Checking the existence of a connected allocation that is (a) \EQ{1} and \NW{}, (b) \EF{1} and \NW{}, (c) \EQ{1} and has utilitarian welfare at least $m$, or (d) \EF{1} and has utilitarian welfare at least $m$ is \NPC{} for a path and a $(4,4)$-sparse binary valuations instance.
\label{cor:EQ1_NW_Egal_Util}
\end{corollary}

\begin{remark}
The hardness result in \Cref{cor:EQ1_NW_Egal_Util} can be extended to \emph{multiplicative} approximations of \EQ{1} and \EF{1}. Given any $\alpha \in [0,1]$, an allocation $A$ is said to satisfy $\alpha$-\EQ{1} (respectively, $\alpha$-\EF{1}) if for every pair of agents $a_i, a_k \in \N$ such that $A_k \neq \emptyset$, there exists some good $v \in A_k$ such that $u_i(A_i) \geq \alpha \cdot u_k(A_k \setminus \{v\})$ (respectively, $u_i(A_i) \geq \alpha \cdot u_i(A_k \setminus \{v\})$).\footnote{Similar approximations have been studied in the context of envy-freeness up to any good (\EFX{})~\citep{PR20almost,AMN20multiple}.} The reasoning is similar: The allocation in the forward direction is \EQ{1} as well as \EF{1}, and hence also $\alpha$-\EQ{1} and $\alpha$-\EF{1}. The argument in the reverse direction only uses non-wastefulness, and therefore vacuously holds for $\alpha$-\EQ{1} (or $\alpha$-\EF{1}). As a result, we obtain that for any rational $\alpha \in [0,1]$, it is \NPC{} to determine the existence of a connected $\alpha$-\EQ{1} (or $\alpha$-\EF{1}) allocation that is non-wasteful or has utilitarian welfare at least $m$.
\end{remark}

A straightforward adaptation of the construction in \Cref{thm:NW_NPcomplete} also gives us the following:
\begin{restatable}[]{theorem}{EQonePO}
Checking the existence of a connected allocation that is (a) \EQ{1} and \PO{}, (b) \EF{1} and \PO{}, (c) \EQ{1} and has egalitarian welfare at least $2$, or (d) \EF{1} and has egalitarian welfare at least $2$ is \NPC{} for a path and a $(6,4)$-sparse binary valuations instance.
\label{thm:EF1_EQ1_PO}
\end{restatable}
The proof of \Cref{thm:EF1_EQ1_PO} is presented in the Appendix \ref{sec:Proof_EF1_EQ1_PO}.

Recently, \citet[Theorem 7]{IP19Pareto} have shown \NPHness{} of checking the existence of a connected \EF{1}+\PO{} allocation of a path even for binary valuations. Their construction involves items that are valued by \emph{all} agents, thus requiring $\O(n)$ sparsity. By contrast, our result in \Cref{thm:EF1_EQ1_PO} shows hardness even for $\O(1)$ sparse instances. Finally, we note that the proof of \Cref{thm:NW_NPcomplete} can also be adapted to show \NPHness{} for egalitarian or utilitarian-optimal \EQ{1} allocations of \emph{chores} (the relevant transformation is $u'_{i,j} = u_{i,j} - 1$).\footnote{For negatively-valued items (or chores), an allocation is said to satisfy \EQ{1} if for every pair of agents $a_i, a_k \in \N$ such that $A_i \neq \emptyset$, there exists a chore $v \in A_i$ such that $v_i(A_i \setminus \{v\}) \geq v_k(A_k)$~\citep{FSV+20equitablechores}.}

\section{Algorithmic Results for Complete \EQ{1} Allocations}
\label{sec:EQ1_Complete_Algorithms}

The intractability results in the previous section
 prompt us to relax the efficiency requirement in search of positive results, and ask the following question: Does there always exist a connected and \emph{complete} \EQ{1} allocation of a path?

A natural approach towards this question is to start with a connected and \emph{exactly} equitable division in a cake-cutting instance derived by relaxing the indivisibility constraint (such divisions are guaranteed to exist~\citep{CDP13existence,AD15efficiency,C17existence}). The fractional cake division could then be rounded to obtain a connected and \emph{approximately} equitable allocation of indivisible goods. 
Unfortunately, there exist instances where \emph{every} rounding of the fractional cake division fails to satisfy \EQ{1}.\footnote{Consider an instance with seven goods $v_1,\dots,v_7$ and three agents with identical valuations $u = (1,1,1,1,1,1,12)$. \emph{Any} connected and equitable division assigns $v_1,\dots,v_6$ to one agent and equally divides $v_7$ between the other two. 
In any rounding, some agent will get an empty bundle, thus violating \EQ{1}.\label{footnote:Cake_rounding_fails_EQ1}}

An alternative approach is to work directly with the indivisible goods instance. For a path graph, any connected allocation can be naturally associated with a left-to-right ordering of agents, say $\sigma$. We call a connected (partial) allocation \emph{$\sigma$-consistent} if it assigns connected bundles from left to right according to $\sigma$. \citet{S19fairly} has shown that there is a polynomial-time local search algorithm that, for any fixed ordering $\sigma$ of agents, finds a connected, complete, $\sigma$-consistent, and approximately equitable allocation. Specifically, his algorithm computes a $u_{\max}$-\EQ{} allocation, where $u_{\max} \coloneqq \max_{a_i \in \N, v \in V} u_i(\{v\})$ is the highest valuation any agent has for any good, and an allocation $A$ is $u_{\max}$-\EQ{} if for every $a_i,a_k \in \N$, we have $|u_i(A_i) - u_k(A_k)| \leq u_{\max}$.

Notice that $u_{\max}$-\EQ{} is a strictly weaker guarantee than \EQ{1}, and there exist instances where Suksompong's algorithm fails to compute an \EQ{1} allocation (even though such an allocation exists).\footnote{Consider the instance in \Cref{footnote:Cake_rounding_fails_EQ1} where $u_{\max}=12$. Starting with the allocation $A \coloneqq (\{\emptyset\}, \{v_1,\dots,v_6\}, \{v_7\})$, Suksompong's local search algorithm immediately returns $A$ as the output since it is $u_{\max}$-\EQ{}, even though it violates \EQ{1}. Observe that the allocation $B \coloneqq (\{v_1,v_2,v_3\},\{v_4,v_5,v_6\}, \{v_7\})$ is \EQ{1} and has a higher egalitarian welfare.} Thus, this approach, too, does not resolve the existence of \EQ{1} and complete allocations. 
Moreover, this algorithm could fail to satisfy standard criteria of \emph{economic efficiency}. Given this limitation, \citet{S19fairly} posed the computation of `approximate equitable allocations with non-trivial welfare guarantees' as an open problem. 

We address this gap by providing a polynomial-time algorithm for computing a connected, complete, and \EQ{1} allocation (\Cref{thm:EQ1+complete-polytime}). Our algorithm also provides the following economic efficiency guarantee: For any given agent ordering $\sigma$, our algorithm returns a connected, $\sigma$-consistent, and \EQ{1} allocation whose egalitarian welfare is the highest among \emph{all} connected and $\sigma$-consistent allocations. In other words, a connected and egalitarian-optimal allocation for any fixed ordering of the agents is, without loss of generality, \emph{fair} (i.e., \EQ{1}) and efficiently computable.

\begin{restatable}[]{theorem}{EQonecomp}
\label{thm:EQ1+complete-polytime}
There is a polynomial-time algorithm for computing a connected, complete, and \EQ{1} allocation of a path consistent with a given ordering of agents. Furthermore, this allocation is egalitarian-optimal among all connected allocations consistent with the given ordering.
\end{restatable}
Note that the strong existence guarantee of \Cref{thm:EQ1+complete-polytime} cannot be extended to \EQ{1} and Pareto optimal allocations. 
 Indeed, consider an instance with five goods and three agents where $u_1 = (1,0,0,0,0)$, $u_2 = (0,1,0,0,0)$, and $u_3 = (0,0,1,1,1)$. For $\sigma=(1,2,3)$, the unique connected, $\sigma$-consistent, and Pareto optimal allocation is $(\{v_1\},\{v_2\},\{v_3,v_4,v_5\})$ which violates \EQ{1}.

\begin{algorithm*}
 \DontPrintSemicolon
 \KwIn{An instance $\I = \langle G, \N, \U \rangle$ and an ordering of agents $\sigma$.}
 \KwOut{A connected allocation $A$.}
 \BlankLine
 \Comment{Phase $1$: Compute the optimal egalitarian welfare $\theta$}
 \BlankLine
 \tikzmk{A}
 $k \leftarrow 1$\;
 $\theta \leftarrow u^k$
 \Comment*[r]{Initialize the guess for optimal egalitarian welfare}
 \While{$k \leq n (m+1)^2$}{
 Starting from the leftmost available good, move rightward along $G$ and tentatively assign a minimal connected bundle worth at least $u^{k+1}$ to each successive agent in $\sigma$.
 \label{algline:EQ1+Comp_Phase1}\;
 \uIf{the assignment in Line~\ref{algline:EQ1+Comp_Phase1} is infeasible\label{algline:EQ1+Comp_Phase1_FeasibilityCheck}}{Exit while-loop and start Phase 2.}
 \Else{
 $\theta \leftarrow u^{k+1}$\Comment*[r]{Update the guess}
 $k \leftarrow k+1$\Comment*[r]{Update $k$}
 }
 }
 \nonl \tikzmk{B}
 \boxit{mygray}
 \Comment{Phase $2$: Find a $\theta$-unsafe agent via a left-to-right scan}
 \BlankLine
 \nl
 \tikzmk{A}
 $A \leftarrow (\emptyset,\dots,\emptyset)$\;
 $i \leftarrow 1$\;
 \While{$i \leq n$}{
 \uIf{there exists a $\sigma$-consistent partial allocation that is identical to $A$ for the agents $a_1,\dots,a_{i-1}$, and, in addition, assigns connected bundles worth at least $u^{k+1}$ to $a_i$ and worth at least $\theta=u^k$ to each of $a_{i+1},\dots,a_n$ \label{algline:EQ1+Comp_Phase2_condition}}
 {$A_i \leftarrow $ the minimal connected bundle worth at least $u^{k+1}$ to agent $a_i$ starting from the leftmost available good.\;
 $G \leftarrow G \setminus A_i$\Comment*[r]{Update the set of remaining goods}
 $i \leftarrow i+1$}
 \Else{Exit while-loop and start Phase 3.\Comment*[r]{$a_i$ is the leftmost $\theta$-unsafe agent}}
 }
 \nonl \tikzmk{B}
 \boxit{mygray}
 \Comment{Phase $3$: Finalize the remaining assignments via a right-to-left scan}
 \BlankLine
 \nl
 \tikzmk{A}
 $k \leftarrow n$\;
 \While{$k > i$}{
 $A_k \leftarrow $ the minimal connected bundle worth at least $\theta=u^k$ to agent $a_k$ starting from the rightmost available good.\label{algline:Minimal_RightToLeft_Assignment}\;
 $G \leftarrow G \setminus A_k$\Comment*[r]{Update the set of remaining goods}
 $k \leftarrow k-1$}
 $A_i \leftarrow G$\Comment*[r]{Assign all remaining goods to the $\theta$-unsafe agent $a_i$}\label{algline:Assign_leftovers_to_unsafe_agent}
 \KwRet{A}\;
 \nonl \tikzmk{B}
 \boxit{mygray}
\caption{Algorithm for finding a connected and complete \EQ{1} allocation.}
 \label{alg:EQ1+Complete}
\end{algorithm*}

\textbf{Description of the algorithm:} Let $\sigma \coloneqq (a_1,a_2,\dots,a_n)$. Our algorithm (see Algorithm~\ref{alg:EQ1+Complete}) consists of three phases. 

In Phase 1, the algorithm computes the \emph{optimal egalitarian welfare} $\theta$ for $\sigma$-consistent allocations. To compute this value, the algorithm starts with a preprocessed list $L = (u^1,u^2,\dots)$ containing all distinct realizable utility values of any agent for any connected bundle, where $u^1 \coloneqq 0 < u^2 < u^3$ and so on. (The list $L$ is of length $\O(nm^2)$ since the number of distinct connected bundles in a path is $\O(m^2)$.) In round $k$, the algorithm checks whether there exists a connected and $\sigma$-consistent partial allocation with egalitarian welfare $u^{k+1}$. To do this, the algorithm starts from the leftmost available good and iteratively assigns \emph{minimal} connected bundles to the agents $a_1,a_2,\dots$ such that each agent gets a utility of at least $u^{k+1}$; here, \emph{minimal} refers to cardinality-wise smallest bundle. If a feasible partial allocation exists, the algorithm updates its `guess' of the achievable egalitarian welfare to $\theta = u^{k+1}$ and moves to round $k+1$. Otherwise, it fixes $\theta=u^k$ and moves to Phase 2. Thus, for the instance in \Cref{fig:theta_unsafe}, the partial allocation in round $1$ ($\theta=0$) is $(\{v_1\},\{v_2,v_3\},\{v_4\})$, and that in round $2$ ($\theta=1$) is $(\{v_1,v_2\},\{v_3,v_4\},\{v_5,v_6,v_7,v_8\})$. In round $3$, the algorithm encounters infeasibility, so it fixes $\theta=2$.
\begin{figure}
\centering
\begin{tikzpicture}
    \tikzmath{\x=0;\edgex=0;\edgewidth=0.25;\nodesep=0.65;\edgenodesep=0.2;\Yoffset=-0.5;\bigradius=0.2;\smallradius=0.05;}
    \draw (\x,0) circle [radius=\bigradius] node [] {\small{$v_1$}};
    \draw (\x,\Yoffset) node [] {\small{$1$}};
    \draw (\x,2*\Yoffset) node [] {\small{$0$}};
    \draw (\x,3*\Yoffset) node [] {\small{$0$}};
    \draw (\x-0.5,\Yoffset) node [] {\small{$a_1:$}};
    \draw (\x-0.5,2*\Yoffset) node [] {\small{$a_2:$}};
    \draw (\x-0.5,3*\Yoffset) node [] {\small{$a_3:$}};
    \draw[fill=black] (\x,5*\Yoffset) circle [radius=\smallradius] node [] {};
    \draw[fill=black] (\x,6.5*\Yoffset) circle [radius=\smallradius] node [] {};
    \tikzmath{\edgex=\x+\edgenodesep;}
    \tikzmath{\x=\x+\nodesep;}
    \draw (\edgex,0) -- (\edgex+\edgewidth,0);
    \draw (\x,0) circle [radius=\bigradius] node [] {\small{$v_2$}};
    \draw (\x,\Yoffset) node [] {\small{$1$}};
    \draw (\x,2*\Yoffset) node [] {\small{$0$}};
    \draw (\x,3*\Yoffset) node [] {\small{$0$}};
    \draw[fill=black] (\x,5*\Yoffset) circle [radius=\smallradius] node [] {};
    \draw[fill=black] (\x,6.5*\Yoffset) circle [radius=\smallradius] node [] {};
    \tikzmath{\edgex=\x+\edgenodesep;}
    \tikzmath{\x=\x+\nodesep;}
    \draw (\edgex,0) -- (\edgex+\edgewidth,0);
    \draw (\x,0) circle [radius=\bigradius] node [] {\small{$v_3$}};
    \draw (\x,\Yoffset) node [] {\small{$1$}};
    \draw (\x,2*\Yoffset) node [] {\small{$1$}};
    \draw (\x,3*\Yoffset) node [] {\small{$0$}};
    \draw[fill=black] (\x,5*\Yoffset) circle [radius=\smallradius] node [] {};
    \draw[fill=black] (\x,6.5*\Yoffset) circle [radius=\smallradius] node [] {};
    \tikzmath{\edgex=\x+\edgenodesep;}
    \tikzmath{\x=\x+\nodesep;}
    \draw (\edgex,0) -- (\edgex+\edgewidth,0);
    \draw (\x,0) circle [radius=\bigradius] node [] {\small{$v_4$}};
    \draw (\x,\Yoffset) node [] {\small{$0$}};
    \draw (\x,2*\Yoffset) node [] {\small{$1$}};
    \draw (\x,3*\Yoffset) node [] {\small{$1$}};
    \draw[fill=black] (\x,5*\Yoffset) circle [radius=\smallradius] node [] {};
    \draw[fill=black] (\x,6.5*\Yoffset) circle [radius=\smallradius] node [] {};
    \tikzmath{\edgex=\x+\edgenodesep;}
    \tikzmath{\x=\x+\nodesep;}
    \draw (\edgex,0) -- (\edgex+\edgewidth,0);
    \draw (\x,0) circle [radius=\bigradius] node [] {\small{$v_5$}};
    \draw (\x,\Yoffset) node [] {\small{$0$}};
    \draw (\x,2*\Yoffset) node [] {\small{$1$}};
    \draw (\x,3*\Yoffset) node [] {\small{$0$}};
    \draw[fill=black] (\x,5*\Yoffset) circle [radius=\smallradius] node [] {};
    \draw[fill=black] (\x,6.5*\Yoffset) circle [radius=\smallradius] node [] {};
    \tikzmath{\edgex=\x+\edgenodesep;}
    \tikzmath{\x=\x+\nodesep;}
    \draw (\edgex,0) -- (\edgex+\edgewidth,0);
        \draw (\x,0) circle [radius=\bigradius] node [] {\small{$v_6$}};
    \draw (\x,\Yoffset) node [] {\small{$0$}};
    \draw (\x,2*\Yoffset) node [] {\small{$1$}};
    \draw (\x,3*\Yoffset) node [] {\small{$1$}};
    \draw[fill=black] (\x,5*\Yoffset) circle [radius=\smallradius] node [] {};
    \draw[fill=black] (\x,6.5*\Yoffset) circle [radius=\smallradius] node [] {};
    \tikzmath{\edgex=\x+\edgenodesep;}
    \tikzmath{\x=\x+\nodesep;}
    \draw (\edgex,0) -- (\edgex+\edgewidth,0);
    \draw (\x,0) circle [radius=\bigradius] node [] {\small{$v_7$}};
    \draw (\x,\Yoffset) node [] {\small{$0$}};
    \draw (\x,2*\Yoffset) node [] {\small{$1$}};
    \draw (\x,3*\Yoffset) node [] {\small{$0$}};
    \draw[fill=black] (\x,5*\Yoffset) circle [radius=\smallradius] node [] {};
    \draw[fill=black] (\x,6.5*\Yoffset) circle [radius=\smallradius] node [] {};
    \tikzmath{\edgex=\x+\edgenodesep;}
    \tikzmath{\x=\x+\nodesep;}
    \draw (\edgex,0) -- (\edgex+\edgewidth,0);
    \draw (\x,0) circle [radius=\bigradius] node [] {\small{$v_8$}};
    \draw (\x,\Yoffset) node [] {\small{$0$}};
    \draw (\x,2*\Yoffset) node [] {\small{$0$}};
    \draw (\x,3*\Yoffset) node [] {\small{$1$}};
    \draw[fill=black] (\x,5*\Yoffset) circle [radius=\smallradius] node [] {};
    \draw[fill=black] (\x,6.5*\Yoffset) circle [radius=\smallradius] node [] {};
\draw (0,5*\Yoffset) coordinate (O1) -- (\x,5*\Yoffset)coordinate(ff);
\draw (-0.1,5*\Yoffset+0.1) rectangle (0.1+4*\bigradius+2*\edgewidth,5*\Yoffset-0.1) node[midway,above=0.08] {\footnotesize{$u_1=3$}};
\draw (-0.1+6*\bigradius+3*\edgewidth,5*\Yoffset+0.1) rectangle (0.1+8*\bigradius+4*\edgewidth,5*\Yoffset-0.1) node[midway,above=0.08] {\footnotesize{$u_2=2$}};
\draw (-0.1+10*\bigradius+5*\edgewidth,5*\Yoffset+0.1) rectangle (0.1+14*\bigradius+7*\edgewidth,5*\Yoffset-0.1) node[midway,above=0.08] {\footnotesize{$u_3=2$}};
\node[] at (\x+1.3,5*\Yoffset) (label1) {\footnotesize{$a_1$ is $2$-safe}};
\node[] at (\x+1.3,3.5*\Yoffset) (label2) {\small{$\theta = u^k = 2$}};
\node[] at (\x+1.3,2.5*\Yoffset) (label2) {\small{$u^{k+1}=3$}};
\draw (0,6.5*\Yoffset) coordinate (O1) -- (\x,6.5*\Yoffset)coordinate(ff);
\draw (-0.1,6.5*\Yoffset+0.1) rectangle (0.1+4*\bigradius+2*\edgewidth,6.5*\Yoffset-0.1) node[midway,above=0.08] {\footnotesize{$u_1=3$}};
\draw (-0.1+6*\bigradius+3*\edgewidth,6.5*\Yoffset+0.1) rectangle (0.1+10*\bigradius+5*\edgewidth,6.5*\Yoffset-0.1) node[midway,above=0.08] {\footnotesize{$u_2=3$}};
\draw[densely dashed] (-0.1+12*\bigradius+6*\edgewidth,6.5*\Yoffset+0.1) rectangle (0.1+14*\bigradius+7*\edgewidth,6.5*\Yoffset-0.1) node[midway,above=0.08] {\footnotesize{$u_3<2$}};
\node[] at (\x+1.3,6.5*\Yoffset) (label1) {\footnotesize{$a_2$ is $2$-\emph{unsafe}}};
\end{tikzpicture}
\caption{Illustrating the notion of $\theta$-unsafe agent on an instance with binary valuations. For $\theta=u^k=2$ and $u^{k+1}=3$, agent $a_1$ is $\theta$-safe because there exists a partial allocation in which $a_1$'s utility is at least $u^{k+1}$ and that of each of its successors is at least $u^k$. Agent $a_2$ is $\theta$-unsafe because giving a utility of at least $u^{k+1}$ to both $a_1$ and $a_2$ necessarily involves $a_3$'s utility being less than $u^k$.}
\label{fig:theta_unsafe}
\end{figure}
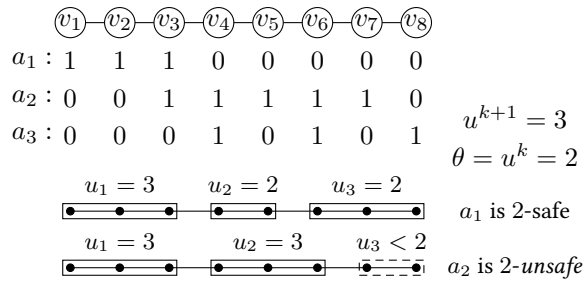

In Phase 2, the algorithm searches for a $\theta$-\emph{unsafe agent}. Given any $\theta = u^k$, we say that agent $a_i$ is $\theta$-\emph{safe} if there exists a connected and $\sigma$-consistent (partial) allocation in which each of $a_1,a_2,\dots,a_i$ gets a utility of at least $u^{k+1}$, and each of $a_{i+1},\dots,a_n$ gets a utility of at least $u^k$. A \emph{$\theta$-unsafe} agent is one that is not $\theta$-safe (see \Cref{fig:theta_unsafe}). Note that a $\theta$-unsafe agent \emph{must exist} since we know from Phase 1 that an egalitarian welfare of $u^{k+1}$ is not possible. The procedure in Phase 1 can be easily adapted to compute the \emph{leftmost} $\theta$-unsafe agent, say $a_i$. Having found $a_i$, the algorithm now fixes the assignments of its predecessors $a_1,\dots,a_{i-1}$ (but not $a_i$) by starting from the leftmost available good and iteratively assigning each agent a minimal connected bundle worth at least $u^{k+1}$. The algorithm now moves to Phase 3.

In Phase 3, the algorithm finalizes the assignments of the remaining agents via a \emph{right-to-left} scan of the path $G$. Specifically, starting from the rightmost available good, the algorithm moves leftwards along $G$ and iteratively assigns minimal connected bundles worth at least $u^k$ to the agents in the reverse order $a_n,a_{n-1},$ and so on. 
Upon encountering $a_i$ for the second time, the algorithm assigns to it all the remaining goods, and returns the final allocation as the output. The proof of \Cref{thm:EQ1+complete-polytime} follows.

\begin{proof} (of \Cref{thm:EQ1+complete-polytime})
We will show that the above algorithm (Algorithm~\ref{alg:EQ1+Complete}) satisfies the desired properties. 

Let us start with the running time analysis of the algorithm assuming that agents have additive valuations (later in \Cref{rem:EQ1_Complete_Beyond_Additive}, we will provide a similar analysis for general monotone valuations). Since there are $\O(m^2)$ possible connected bundles for each of the $n$ agents, the computation of $L$ requires computing the utility of $\O(nm^2)$ bundles by adding up the individual utilities of the constituent goods. This amounts to a total running time of $\O(nm^2\log(u_{\max}))$ for the preprocessing phase. Further, we assume that the agents' utilities for all possible bundles are cached so as to facilitate constant time access in the remainder of the algorithm.

The total running time for Phase 1 is $\O(nm^3)$, since there are at most $n(m+1)^2$ iterations of the while-loop and each iteration involves scanning at most $m$ goods. 
Phase 2 involves at most $n$ iterations of the while-loop, and each iteration requires assigning connected bundles to all agents from left to right, which takes $\O(m)$ time. Thus, the total running time for Phase 2 is $\O(nm)$. 
In Phase 3, each good is considered at most once during the right-to-left scan, resulting in a running time of $\O(m)$. Thus, overall, the algorithm takes $\O(nm^3 + nm^2\log(u_{\max}))$ time.

The allocation $A$ returned by the algorithm is \emph{complete} because all leftover goods are allocated in the last step, and is \emph{$\sigma$-consistent} because this property is maintained by the algorithm at every step. Furthermore, $A$ is also \emph{connected} since the algorithm assigns connected bundles to $a_1,\dots,a_{i-1}$ from left to right and to $a_n,\dots,a_{i+1}$ from right to left. (The feasibility of the right-to-left assignment is guaranteed by the fact that $a_{i-1}$ is $\theta$-safe, as $a_i$ is \emph{leftmost} $\theta$-unsafe agent.) Since $G$ is a path, the set of leftover goods assigned to $a_i$ in Phase 3 is also connected.

We will now argue that $A$ is \emph{egalitarian-optimal} among all connected and $\sigma$-consistent allocations. First, observe that the value $\theta = u^k$ fixed at the end of Phase 1 is indeed the \emph{optimal} egalitarian welfare of any connected and $\sigma$-consistent allocation. (Otherwise, by monotonicity of valuations, there must exist a connected and $\sigma$-consistent allocation with egalitarian welfare $u^{k+1}$ or higher in which agents receive \emph{minimal} bundles. This, however, would contradict the infeasibility encountered for $\theta=u^{k+1}$.) Next, we will show that the egalitarian welfare of $A$ is equal to $u^k$, which will establish egalitarian-optimality. Indeed, each of $a_1,\dots,a_{i-1}$ gets a utility of at least $u^{k+1} > u^k$ in Phase 2, and each of $a_n,\dots,a_{i+1}$ gets a utility of at least $u^k$ in Phase 3. The utility of $a_i$ for its assigned bundle is \emph{exactly} $u^k$ because of the following two reasons: First, $a_i$'s utility is at least $\theta = u_k$ since $a_{i-1}$ is $\theta$-safe (recall that $a_i$ is the \emph{leftmost} $\theta$-unsafe agent). Second, since $a_i$ is $\theta$-unsafe, assigning a bundle worth at least $u^{k+1}$ to $a_i$ (and each of its predecessors $a_1,\dots,a_{i-1}$) would imply that one of its successors $a_{i+1},\dots,a_n$ gets utility strictly below $u^k$, which contradicts the assignments in Phase 3. Thus, $a_i$'s utility must be strictly below $u^{k+1}$, and hence, equal to $u^k = \theta$.

Finally, to prove that $A$ is \EQ{1}, notice that if the utility of an agent is strictly greater than $u^k$ (in particular, each of $a_1,\dots,a_{i-1}$ gets a utility at least $u^{k+1} > u^k$), then by \emph{minimality} of bundles, there must exist a boundary good whose removal results in the agent's residual utility being strictly below $u^{k+1}$, and therefore less than or equal to $u^k$. Since each agent gets a utility at least $u^k$, $A$ must be \EQ{1}.
\end{proof}

We observe that the running time of Algorithm~\ref{alg:EQ1+Complete} can be improved to $\OO(nm^{2})$ via following modifications: In the preprocessing step, as before, we go through all possible $\OO(nm^{2})$ bundles. We cache these bundles, and compute $u_{\max}$ which amounts to the running time of $O(nm^{2})$ for this step. Next, in Phase 1, we use binary instead of the linear search to find the \emph{optimal egalitarian welfare} $\theta$. With these modifications, Phase 1 runs in time $\OO(m \log mn)$. In Phase 2, the $\theta$-unsafe agent can be found in $\OO(m)$ time with a combination of left-to-right scan that tentatively assigns bundles worth $u^{k+1}$ and a right-to-left scan that assigns bundles worth $u^{k}$. Finally, Phase 3 runs takes $\OO(m)$ time as before.  

\begin{remark}
Note that the algorithm in \Cref{thm:EQ1+complete-polytime} and the analysis of its correctness only use the monotonicity of valuations, and therefore the result extends to \emph{non-additive} utilities. The running time analysis in this case relies on the existence of a valuation oracle that, given a connected bundle, returns the agent's utility for that bundle. Since the number of distinct connected bundles in a path is $\O(m^2)$, after $\O(nm^2)$ valuation queries, each agent's value for every connected bundle is available to the algorithm. The rest of the analysis is identical to that in \Cref{thm:EQ1+complete-polytime}.
\label{rem:EQ1_Complete_Beyond_Additive}
\end{remark}

\begin{remark}
Another relevant implication is that our algorithm can be easily adapted for negative valuations to obtain the efficient computation of connected \EQ{1} allocations for \emph{chores}. The latter result provides a tractable alternative to a recent result showing \NPHness{} for connected and exactly equitable chore allocations~\citep{BCL19chore}.
\end{remark}

\begin{remark}
\citet{BCF+19almost} and \citet{OPS19fairly} have independently shown that when agents have identical monotone valuations, a connected \EF{1} allocation of a path can be efficiently computed. Since \EF{1} and \EQ{1} coincide for identical valuations, our result in \Cref{thm:EQ1+complete-polytime} implies this result as a corollary.

Additionally, we note that although the algorithm of \citet{BCF+19almost} and its analysis are presented for identical valuations, a natural extension of their algorithm for general valuations can be used to derive an alternative proof of \Cref{thm:EQ1+complete-polytime}.
\label{rem:EF1_identical}
\end{remark}

\subsection{\EQ{1} and \PO{}* Allocations}
\label{subsec:EQ1_PO*_Algorithm}

The existence result in \Cref{thm:EQ1+complete-polytime} is quite general, since it applies to any fixed ordering of agents and any monotone valuations instance, and reconciles fairness (i.e., \EQ{1}) with a weak form of economic efficiency (i.e., completeness). On closer inspection, though, we find that it implies an even stronger existence result. Specifically, given an agent ordering $\sigma$, let $\mathcal{A}^\sigma$ denote the set of all connected, $\sigma$-consistent, \EQ{1} and complete allocations for the given instance. From \Cref{thm:EQ1+complete-polytime}, we know that $\mathcal{A}^\sigma$ is non-empty. Furthermore, since there are only finitely many allocations, there must exist an allocation in $\mathcal{A}^\sigma$ that is not Pareto dominated by any other allocation in $\mathcal{A}^\sigma$. We call this property \PO{}*.

Formally, given an agent ordering $\sigma$, we say that allocation $A$ is \PO{}* if it is connected, $\sigma$-consistent, complete, and \EQ{1}, and no other connected, $\sigma$-consistent, complete, and \EQ{1} allocation Pareto dominates $A$. From the aforementioned argument, it follows that a \PO{}* allocation always exists.

Intriguingly, while Algorithm~\ref{alg:EQ1+Complete} can be used to establish the existence of a \PO{}* allocation even for general monotone valuations, it can fail to return such an allocation even for binary additive valuations. Indeed, consider an instance with five goods $v_1, \dots v_5$ and three agents with valuations $u_1 = (1,0,0,1,0)$, $u_2 = (0,1,1,0,0)$, and $u_3 = (0,0,0,1,1)$. Given the ordering $\sigma=(1,2,3)$, Algorithm~\ref{alg:EQ1+Complete} computes a $\sigma$-consistent and \EQ{1} allocation $(\{v_1, v_2\}, \{v_3, v_4\}, \{v_5\})$ with utility profile $(1, 1, 1)$, which is Pareto dominated by another $\sigma$-consistent and \EQ{1} allocation $(\{v_1\}, \{v_2, v_3\}, \{v_4, v_5\})$ with utility profile $(1, 2, 2)$.

Thus, for a given agent ordering $\sigma$, \PO{}* is stronger than \EQ{1}+completeness as in this case, the former implies the latter. We note that \PO{}* with respect to an ordering $\sigma$ could be weaker (i.e., Pareto dominated) than an \EQ{1}+complete allocation with respect to a different ordering $\sigma'$. 


This motivates the following natural question: Given an ordering $\sigma$, can a \PO{}* allocation be efficiently computed? While we are unable to settle this question for general monotone valuations, in \Cref{thm:EQ1+POstar-polytime} we show that a variant of Algorithm~\ref{alg:EQ1+Complete} efficiently computes a \PO{}* allocation for binary additive valuations.

\begin{restatable}[]{theorem}{EQonePOstarcomp}
\label{thm:EQ1+POstar-polytime}
Given an instance with binary additive valuations and any agent ordering $\sigma$, a connected, $\sigma$-consistent, \EQ{1}, and \PO{}* allocation of a path can be computed in polynomial time. 
\end{restatable}

The proof of \Cref{thm:EQ1+POstar-polytime} is presented in the Appendix \ref{sec:Proof_EQ1+POstar-polytime}.

\section{Structured Preferences}

In this section, we will explore a different avenue for circumventing the intractability associated with non-wasteful \EQ{1} allocations. 
Unlike in \Cref{thm:EQ1+complete-polytime} 
 where we relaxed the efficiency requirement, this time we will instead assume that agents have \emph{structured} preferences. In particular, we will focus on 
\emph{binary extremal valuations} wherein for each agent $a_i$, either there exists $\ell_i \in [m]$ such that $u_{i,j} = 1$ for all $j \in \{1,\dots,\ell_i\}$ and $0$ otherwise (i.e., $a_i$ is \emph{left-extremal}), or there exists $r_i \in [m]$ such that $u_{i,j} = 1$ for all $j \in \{r_i,\dots,m\}$ and $0$ otherwise (i.e., $a_i$ is \emph{right-extremal}). Similar domain restrictions have been previously considered in the context of voting problems~\citep{EL15structure}.

\begin{restatable}[]{theorem}{EQoneNWExtremal}
 \label{thm:EQoneNW_Extremal}
 There is a polynomial-time algorithm that, given an instance with binary extremal and additive valuations, returns a connected, non-wasteful, and \EQ{1} allocation whenever such an allocation exists.
\end{restatable}
\begin{proof}
We will show that the desired allocation, if it exists, can be obtained by concatenating the solutions from two subproblems, one on a purely left-extremal and the other on a purely right-extremal subinstance.

Suppose there exists a connected, non-wasteful (\NW{}), and \EQ{1} allocation $A$. Let $\sigma$ denote the agent ordering under $A$. By relabeling the agents, we have that $\sigma=(a_1,\dots,a_n)$. We claim that without loss of generality, all left-extremal agents precede all right-extremal agents in $\sigma$. Indeed, if there is a pair of adjacent agents $a_i,a_{i+1}$ where $a_i$ is right-extremal and $a_{i+1}$ is left-extremal, then by an exchange argument we can obtain another connected, non-wasteful, and \EQ{1} allocation $B$ where such a violation does not occur. Specifically, by swapping the bundles of $a_i$ and $a_{i+1}$, we maintain connectedness and non-wasteful. Additionally, for binary additive valuations, non-wastefulness implies that the utility of an agent is equal to the cardinality of its bundle. Therefore, swapping bundles results in swapping the utility values of $a_i$ and $a_{i+1}$, which means that the old and new allocations have identical utility profiles (up to relabeling). Thus, allocation $B$ must also satisfy \EQ{1}.

Let $v_j \in V$ be such that the set $V^L \coloneqq \{v_1,\dots,v_j\}$ is allocated among the left-extremal agents and $V^R \coloneqq \{v_{j+1},\dots,v_m\}$ is allocated among the right-extremal agents in $A$. Then, the subinstance restricted to $V^L$ only has left-extremal valuations and admits a connected, non-wasteful, and \EQ{1} allocation (indeed, the restriction of $A$ to $V^L$ satisfies these properties). A similar implication holds for the purely right-extremal subinstance $V^R$. Therefore, it suffices to provide a polynomial-time algorithm for checking the existence of a connected, non-wasteful, and \EQ{1} allocation in a binary left-extremal instance. Notice that the same algorithm can be used for the right-extremal subinstance via an easy `mirror transformation'. If both subinstances admit desired allocations, then the concatenated allocation is clearly connected and non-wasteful in the original instance. By checking this allocation for \EQ{1}, we obtain the desired algorithm for the original instance. 
Thus, in rest of the proof, we will focus only on left-extremal valuations.

Let $A'$ denote the restriction of allocation $A$ to the left-extremal subinstance, and let $n'$ and $m'$ correspondingly denote the number of agents and items, respectively. Since $A'$ is non-wasteful and \EQ{1} and the valuations are binary, the minimum and maximum utilities under $A'$ must be $\lfloor \frac{m'}{n'} \rfloor$ and $\lceil \frac{m'}{n'} \rceil$, respectively. That is, an agent is either a `floor' or a `ceiling' agent. By an exchange argument, it can be shown that for any pair of left-extremal agents $a_i,a_k$ such that $i < k$, we have $\ell_i \leq \ell_k$ (i.e., $a_i$'s interval finishes before $a_k$'s) without loss of generality. Similarly, it holds that the floor agents precede the ceiling agents (here, the exchange argument transfers a boundary item).

Let $n'_f$ and $n'_c$ denote the number of floor and ceiling agents, respectively. Thus, $n'_f$ and $n'_c$ are the unique pair of non-negative integers satisfying the equations $n'_f + n'_c = n' \text{ and } \displaystyle n'_f \cdot \floor{\frac{m'}{n'}} + n'_c \cdot \ceil{\frac{m'}{n'}} = m'$. (If $m' = kn'$ for some $k \in \mathbb{N}$, then $n'_f = n'$ and $n'_c = 0$.) The desired algorithm considers the agents in the order in which their intervals finish, and constructs an allocation as follows: Starting from the leftmost available good, the algorithm assigns a connected bundle of $\floor{\frac{m'}{n'}}$ goods to each of the first $n'_f$ agents, and a connected bundle of $\ceil{\frac{m'}{n'}}$ goods to each of the next $n'_c$ agents. If this allocation is non-wasteful and \EQ1{}, then the algorithm reports YES and returns the said allocation, otherwise it reports NO.
\end{proof}

\section{Concluding Remarks}
We initiated the study of \EQ{1} allocations under connectedness constraints. The pursuit of connected \EQ{1} allocations satisfying non-trivial efficiency guarantees resulted in computational hardness. This result motivated the exploration of two avenues for tractability: relaxing the efficiency requirement and assuming structured preferences. Some of our results found broader applicability to other fairness notions (e.g., \EF{1}) and negatively-valued items.

Going forward, it would be very interesting to explore the domain of \emph{binary intervals} without the extremal structure in search of tractability results. Another relevant direction could be to map the intractability frontier for binary valuations in terms of \emph{$(a,b)$-sparsity}. Our results establish hardness of a number of problems even under $(4,4)$-sparsity. On the other hand, $(1,b)$-sparse instances are efficiently solvable for any $b$. Resolving the complexity of intermediate cases is a natural next step. 
Finally, extensions to \emph{general graphs}~\citep{BCE+17fair} or settings with \emph{mixed items} involving goods as well as chores~\citep{ACI+19fair} could also be of interest.

\section*{Acknowledgments}
We thank the anonymous reviewers for helpful comments and suggestions. NM is supported by the SERB ECR grant ECR/2018/002967, Computational Aspects of Social Choice: Theory and Practice. RV acknowledges support from ONR\#N00014- 171-2621 while he was affiliated with Rensselaer Polytechnic Institute, and is currently supported by project no. RTI4001 of the Department of Atomic Energy, Government of India. Part of this work was done while RV was supported by the Prof. R Narasimhan postdoctoral award.

\bibliographystyle{abbrvnat} 
\bibliography{references}

\clearpage
\begin{appendices}

\section{Proof of Theorem~\ref{thm:EF1_EQ1_PO}}
\label{sec:Proof_EF1_EQ1_PO}

Let us recall the statement of \Cref{thm:EF1_EQ1_PO}.

\EQonePO*

We will start by discussing the proof of part (a) of \Cref{thm:EF1_EQ1_PO}, followed by that of part (c) which uses the same construction. Parts (b) and (d) use a slightly different construction, and their proofs will be presented subsequently.

\begin{proof} (of part (a))
We will show a reduction from \LESfull{} (\LES{}) and our construction will be similar to that of \Cref{thm:NW_NPcomplete}. Recall that an instance of \LES{} consists of $5p$ clauses (where $p \in \mathbb{N}$) denoted as follows:
$$ \CC = \{U_1, V_1, U_1^\pr, V_1^\pr, \cdots, U_p, V_p, U_p^\pr, V_p^\pr\} \cup \{ C_1, \cdots, C_p \}.$$
We will refer to the first $4p$ clauses as the \emph{core} clauses, and the remaining clauses as the \emph{auxiliary} clauses. The set of variables consists of $p$ \emph{main variables} $x_1,\dots,x_p$ and $4p$ \emph{shadow variables} (our notation for the shadow variables will differ slightly from that used in \Cref{thm:NW_NPcomplete}). 
Each core clause consists of two literals and has the following structure:
$$\forall \, i \in [p], U_i^{} \cap V_i^{} = \{x_i\} \mbox{ and } U_i^\pr \cap V_i^\pr = \{\bar{x}_i\}.$$

Each main variable $x_i$ occurs exactly twice as a positive literal and exactly twice as a negative literal. The main variables only occur in the core clauses. Each shadow variable makes two appearances: as a positive literal in an auxiliary clause and as a negative literal in a core clause. For $i \in [p]$, we will let $p_{i}$, $r_{i}$, $q_{i}$, and $s_{i}$ denote the shadow variables that appear (as negative literals) in the core clauses $U_i$, $V_i$, $U'_i$ and $V'_i$, respectively. That is, $U_{i} := (\bar{p_{i}} \wedge x_{i}), V_{i} := (\bar{r_{i}} \wedge x_{i})$, $U_{i}' := (\bar{q_{i}} \wedge \bar{x_{i}})$, and $V_{i}' := (\bar{s_{i}} \wedge \bar{x_{i}})$. 
Each auxiliary clause consists of four literals, each corresponding to a positive occurrence of a shadow variable. 

The \LES{} problem asks whether, given a set of clauses with the aforementioned structure, there exists an assignment $\tau$ of truth values to the variables such that \emph{exactly one} literal in every core clause and \emph{exactly two} literals in every auxiliary clause evaluate to \true{} under $\tau$.

\textbf{Construction of the reduced instance.}
Let $\phi$ be an instance of LNES. We will begin with the description of the reduced instance.

\textbf{Goods:} For every $i \in [p]$, we introduce one good for every core clause denoted by $U_i^{}$, $V_i^{}$, $U_i^\pr$, $V_i^\pr$, and six goods for every auxiliary clause denoted by $C_{i}^{L_{1}}, C_{i}^{L_{2}}, S_{i}^{1}, S_{i}^{2}, C_{i}^{R_{1}}, C_{i}^{R_{2}}$. We refer to $U_i^{}$, $V_i^{}$, $U_i^\pr$, $V_i^\pr$ as the \emph{core} goods, $C_{i}^{L_{1}}, C_{i}^{L_{2}}, C_{i}^{R_{1}}, C_{i}^{R_{2}}$ as the \emph{auxiliary} goods, and $S_{i}^{1}, S_{i}^{2}$ as the \emph{separator goods}. 
Next, we introduce two goods for each shadow variable, i.e., corresponding to each of $p_{i}, q_i, r_{i}, s_{i}$, we introduce the following \emph{shadow} goods: $p_{i}^{1}, p_{i}^{2}, r_{i}^{1}, r_{i}^{2}, q_{i}^{1}, q_{i}^{2}, s_{i}^{1}, s_{i}^{2}$. Finally, we introduce $2p$ \emph{dummy} goods denoted by $D_1,D_1^{'}, \ldots, D_p, D_{p}^{'}$, two additional \emph{separator} goods $S_{0}^{1}, S_{0}^{2}$, and three \emph{special} goods $S_{1}, S_{2}, S_{3}$. Thus, the total number of goods is $m = 4p+6p+8p+2p+2+3 = 20p+5$. The goods are arranged as shown in \Cref{fig:Thm2-EQ1-PO}.

\begin{figure*}[t]
\centering
    $U_{1}$, $p_{1}^{1}$, $p_{1}^{2}$, $r_{1}^{1}$, $r_{1}^{2}$, $V_{1}$, $U_{1}'$, $q_{1}^{1}$, $q_{1}^{2}$, $s_{1}^{1}$, $s_{1}^{2}$, $V_{1}'$, $\cdots$, $U_{p}$, $p_{p}^{1}$, $p_{p}^{2}$, $r_{p}^{1}$, $r_{p}^{2}$, $V_{p}$, $U_{p}'$, $q_{p}^{1}$, $q_{p}^{2}$, $s_{p}^{1}$, $s_{p}^{2}$, $V_{p}'$\\
    
    (Core and shadow goods)\\
    \vspace{0.1in}
    %
    %
    $S_{0}^{1}$, $S_{0}^{2}$, $C_{1}^{L_{1}}$, $C_{1}^{L_{2}}$, $S_{1}^{1}$, $S_{1}^{2}$, $C_{1}^{R_{1}}$, $C_{1}^{R_{2}}$, $\cdots$, $C_{p}^{L_{1}}$, $C_{p}^{L_{2}}$, $S_{p}^{1}$, $S_{p}^{2}$, $C_{p}^{R_{1}}$, $C_{p}^{R_{2}}$\\
    
    (Separator and auxiliary goods)\\
    \vspace{0.1in}
    %
    %
    $D_{1}$, $D_{1}^{'}$, $D_{2}$, $D_{2}^{'}$, $\cdots$, $D_{p}$, $D_{p}^{'}$, $S_{1}$, $S_{2}$, $S_{3}$\\
    (Dummy and special goods)
    \caption{The instance used in the proof of part (a) of \Cref{thm:EF1_EQ1_PO}. The path graph is such that the goods in the top row are to the left of those in the middle row, which are to the left of those in the bottom row.}
    \label{fig:Thm2-EQ1-PO}
\end{figure*}

\textbf{Agents:} For every main variable $x_i$, we will introduce two agents $a_{x_i}$ and $a_{\bar{x}_i}$ for the two literals; these are referred to as \emph{main agents} of the \emph{positive} and \emph{negative} type, respectively. For every $i \in [p]$, the agent $a_{x_i}$ approves (i.e., values at $1$) the goods $U_i^{}, V_i^{}, D_i^{}, D_i^\pr$, while the agent $a_{\bar{x_i}}$ approves the goods $U_i^\pr, V_i^\pr, D_i^{}, D_i^\pr$. 
We also introduce a \emph{shadow agent} for every shadow variable. 
If $p_{i}$ is a shadow variable occurring in core clause $U_{i}$ and auxiliary clause $C_{j}$, then the corresponding shadow agent $p_{i}$ approves the shadow goods $p_{i}^{1}, p_{i}^{2}$ and the auxiliary goods $C_{j}^{L_{1}}, C_{j}^{L_{2}}, C_{j}^{R_{1}}, C_{j}^{R_{2}}$. The valuations of the other shadow agents $r_{i}, q_{i}, s_{i}$ are defined analogously. Next, we introduce $p+1$ \emph{separator agents} $t_{0}, \dots, t_{p}$ such that for every $i \in \{0\} \cup [p]$, $t_{i}$ approves two separator goods $S_{i}^{1}, S_{i}^{2}$. Lastly, we introduce \emph{special agent} $a_{s}$ that approves the special goods $S_{1}, S_{2}, S_{3}$.

This completes the construction of our reduction. Notice that the constructed instance is $(6, 4)$\emph{-sparse}. Before presenting the proof of equivalence, we will establish in \Cref{lem:EQ1+PO-utilities} that each agent (except for the special agent) has a utility of $2$ under any \EQ{1} and Pareto optimal allocation.

\begin{lemma}
\label{lem:EQ1+PO-utilities}
In any $\EQ{1}+\PO{}$ allocation, the utility of the special agent $a_s$ is equal to $3$ and that of every other agent is equal to $2$.
\end{lemma}
\begin{proof} (of \Cref{lem:EQ1+PO-utilities})
Notice that in any Pareto optimal allocation $A$, the special goods $S_{1}, S_{2}, S_{3}$ must be allocated to the special agent $a_{s}$. This is because these goods lie at the end of the path and are uniquely valued by $a_{s}$, and therefore any allocation $A'$ that does not assign these goods to $a_s$ can be shown to be Pareto dominated by another allocation that is identical to $A'$ except for the assignment of the special goods to the special agent. Therefore, the utility of $a_s$ under Pareto optimal allocation must be equal to $3$ (recall that $a_s$ does not value any good other than the special goods). 

Now let $A$ denote any \EQ{1} and Pareto optimal allocation. Since the utility of the special agent in $A$ is equal to $3$, \EQ{1} requires that the utility of every other agent in $A$ is at least $2$.

Since each separator agent $t_{0},t_1,\dots,t_{p}$ approves exactly two goods, it must be that for every $i \in \{0,1,\dots,p\}$, the separator goods $S_i^1,S_i^2$ are assigned to $t_i$ in $A$. Furthermore, since the separator goods $S_i^1,S_i^2$ are placed next to each other on the path and these are the only goods approved by $t_i$, we can assume, without loss of generality, that these are the only goods assigned to $t_i$.

Now consider a shadow agent $p_{i}$ that appears in the core clause $U_{i}$ and the auxiliary clause $C_{j}$. Thus, $p_{i}$ approves two shadow goods $p_{i}^{1}, p_{i}^{2}$ and four auxiliary goods $C_{j}^{L_{1}},C_{j}^{L_{2}},C_{j}^{R_{1}},C_{j}^{R_{2}}$. Note that $p_i$ cannot receive more than two approved goods; if it does, then by connectedness constraint, its bundle should necessarily include separator goods whose assignment has already been fixed. Thus, each shadow agent $p_i$ (analogously $q_i$, $r_i$, $s_i$) will have a utility of exactly $2$ in $A$.

A similar argument shows that for any $i \in [p]$, the main agent of positive (or negative) type $a_{x_{i}}$ (or $a_{\bar{x_{i}}}$) will have a utility of at most $2$ since all such agents approve two core goods and two dummy goods. We therefore have that in any \EQ{1} and Pareto optimal allocation, all agents other than the special agent achieve a utility of exactly $2$. This completes the proof of \Cref{lem:EQ1+PO-utilities}.
\end{proof}

\textbf{The Forward Direction.} 
Given a satisfying assignment $\tau$ for \LES{}, we will construct the desired allocation as follows:
\begin{itemize}
    \item Allocate the special goods $S_{1}, S_{2}, S_{3}$ to the special agent $a_{s}$.
    
    \item For each $i \in \{0,1,\dots,p\}$, the  separator agent $t_i$ receives the separator goods $S_{i}^{1}$ and $S_{i}^{2}$.
    
    \item If $\tau(x_{i}) = 1$, then allocate $\{U_{i}, p_{i}^{1}, p_{i}^{2}, r_{i}^{1}, r_{i}^{2}, V_{i}\}$ to agent $a_{x_{i}}$ and $\{D_{i}, D_{i}'\}$ to agent $a_{\bar{x}_{i}}$. In addition, allocate $\{U_{i}', q_{i}^{1},q_{i}^{2}\}$ to $q_{i}$, and $\{s_{i}^{1}, s_{i}^{2}, V_{i}'\}$ to $s_{i}$. Recall that $q_{i}$ and $s_{i}$ are the shadow variables that appear as negated literals in the core clauses $U_{i}'$ and $V_{i}'$, respectively, along with $\bar{x_{i}}$.
    
    Otherwise, if $\tau(x_{i}) = 0$, then allocate $\{U'_{i}, q_{i}^{1}, q_{i}^{2}, s_{i}^{1}, s_{i}^{2}, V'_{i}\}$ to agent $a_{\bar{x}_{i}}$ and $\{D_{i}, D_{i}'\}$ to agent $a_{x_{i}}$. In addition, allocate $\{U_{i}, p_{i}^{1},p_{i}^{2}\}$ to $p_{i}$, and $\{r_{i}^{1}, r_{i}^{2}, V_{i}\}$ to $r_{i}$.
    
    \item Finally, for every $j \in [p]$, allocate the sets $\{C_{j}^{L_{1}}, C_{j}^{L_{2}}\}$ and $\{C_{j}^{R_{1}}, C_{j}^{R_{2}}\}$ to the two shadow agents whose corresponding literals satisfy the auxiliary clause $C_j$.
    
\end{itemize}

Observe that each good is assigned to exactly one agent in the aforementioned allocation. Furthermore, each agent's bundle is connected; in particular, each shadow agent either receives a set of adjacent core and shadow goods (if the corresponding shadow variable evaluates to false under $\tau$), or a set of adjacent auxiliary goods (if it evaluates to true).

It is easy to verify that the utility of the special agent is equal to $3$, and that of every other agent is equal to $2$. Thus, the allocation is \EQ{1}.

We will now argue that the above allocation, say $A$, is Pareto optimal. Suppose, for contradiction, that another allocation $A'$ Pareto dominates $A$. Since the special agent and each separator agent receives all of its approved goods under $A$, the utilities of these agents under $A$ and $A'$ must be equal. Furthermore, if a main agent has a strictly higher utility under $A'$, then by the connectedness constraint, its bundle must contain a separator good, which leads to an infeasible assignment since these goods are necessarily allocated to the separator agents. A similar argument shows that a shadow agent, too, cannot receive a higher utility under $A'$. Therefore, $A$ must be Pareto optimal.

\textbf{The Reverse Direction.} We will now show how to recover an \LES{} assignment given a connected \EQ{1} and Pareto optimal allocation, say $A$. 

Since $A$ is \EQ{1} and Pareto optimal, we know from \Cref{lem:EQ1+PO-utilities} that the special agent receives three approved goods and every other agent receives two approved goods under $A$. Thus, in particular, 
for every $i \in \{0,1,\dots,p\}$, the separator goods $S_i^1,S_i^2$ are allocated to the separator agent $t_i$. Along with the connectedness constraint, this implies that for every $i \in [p]$, at least one of the main agents $a_{x_i}$ or $a_{\bar{x}_i}$ will achieve a utility of $2$ by either receiving the interval $U_{i}, p_{i}^{1}, p_{i}^{2}, r_{i}^{1}, r_{i}^{2}, V_{i}$ or $U_{i}', q_{i}^{1}, q_{i}^{2}, s_{i}^{1}, s_{i}^{2}, V_{i}'$. This, in turn, forces \emph{at least} one pair of shadow agents---either $\{p_{i}, r_{i}\}$ or $\{q_{i}, s_{i}\}$---to obtain their utilities from the auxiliary goods. 

We will now show that \emph{exactly} one of these two pairs of agents derive their utility from the shadow goods, while the other pair meets the utility requirement though the auxiliary goods. Indeed, since there are $4p$ auxiliary goods (corresponding to $p$ auxiliary clauses), at most $2p$ shadow agents can obtain the desired utility from the auxiliary goods. Therefore, for every $i \in [p]$, exactly one pair of shadow agents---either $\{p_{i}, r_{i}\}$ or $\{q_{i}, s_{i}\}$---are assigned shadow goods, while the other pair receives auxiliary goods. Note that this observation also shows that for every $i \in [p]$, exactly one out of $a_{x_{i}}$ or $a_{\bar{x_{i}}}$ is assigned the dummy goods $\{D_{i}, D_{i}'\}$.

Overall, we have that one set of $p$ main agents gets exactly two core goods each (we will refer them as the ``lucky'' agents), while the other set of $p$ main agents gets two dummy goods each (the ``unlucky'' agents). Notice that the two main agents corresponding to a main variable cannot both be lucky, nor can both be unlucky due to the argument presented earlier.

This brings us to a natural way of deriving an \LES{} assignment $\tau$ from the allocation $A$. If the main agent of the positive (respectively, negative) type is unlucky, then we let $\tau(x_i) = 0$ (respectively, $\tau(x_i) = 1$). Furthermore, if $A$ allocates a core good to a shadow agent, then the corresponding shadow variable is set to $0$, while shadow variables corresponding to shadow agents who receive auxiliary goods are set to $1$. Note that exactly $2p$ of the $4p$ shadow variables are set to $1$ under this assignment and there are no conflicting assignments, implying that $\tau$ is indeed a valid solution to the \LES{} instance. This completes the proof of part (a) of \Cref{thm:EF1_EQ1_PO}.
\end{proof}

\begin{proof} (of part (c))
To prove part (c), we first observe that the argument in the forward direction remains the same as in part (a), since the allocation constructed in the proof is \EQ{1} and satisfies the desired egalitarian welfare condition. 

In the reverse direction, it is possible that under the given allocation, say $A$, the special agent $a_s$ no longer receives all three special goods. However, since the egalitarian welfare of $A$ is at least $2$, each agent must receive at least two approved goods. Along with connectedness, this means that either $S_1$ or $S_3$ is not assigned to $a_s$ under $A$. Since the special goods are not approved by any other agent, we can modify $A$ to obtain another allocation, say $A'$, that is identical to $A$ except for the allocation of the special goods, which are all assigned to the special agent. It is easy to see that $A'$ is connected, \EQ{1}, and has egalitarian welfare at least $2$. By an identical argument as in part (a), we can now infer a satisfying \LES{} assignment.
\end{proof}

We now move on to the proof of part (b) of \cref{thm:EF1_EQ1_PO}, followed by that of part (d) which uses a similar construction. 

\begin{proof} (of part (b))
We will once again show a reduction from \LESfull{} (\LES{}).

\textbf{Construction of the reduced instance.} Let $\phi$ be an instance of \LES{}. We will begin with the description of the reduced instance.

\textbf{Goods:} For every $i \in [p]$, we introduce one \emph{core} good for every core clause denoted by $U_i^{}$, $V_i^{}$, $U_i^\pr$, $V_i^\pr$, and two \emph{auxiliary} goods for every auxiliary clause denoted by $C_{i}^{L}, C_{i}^{R}$.  
Next, we introduce two goods for each shadow variable, i.e., corresponding to each of $p_{i}, q_i, r_{i}, s_{i}$, we introduce the \emph{shadow} goods $p_{i}^{1}, p_{i}^{2}, r_{i}^{1}, r_{i}^{2}, q_{i}^{1}, q_{i}^{2}, s_{i}^{1}, s_{i}^{2}$. Finally, we introduce $3p$ \emph{dummy} goods 
 $D_1^1,D_1^2,D_1^3, \dots, D_p^1,D_p^2,D_p^3$ and two \emph{separator} goods $S_{0}^{1}, S_{0}^{2}$. Thus, the total number of goods is $m = 4p+2p+8p+3p+2 = 17p+2$. The goods are arranged as shown in \Cref{fig:EF1_PO_appendix}.
\begin{figure*}[t]
\centering
    $U_{1}$, $p_{1}^{1}$, $p_{1}^{2}$, $r_{1}^{1}$, $r_{1}^{2}$, $V_{1}$, $U_{1}'$, $q_{1}^{1}$, $q_{1}^{2}$, $s_{1}^{1}$, $s_{1}^{2}$, $V_{1}'$, $\cdots$, $U_{p}$, $p_{p}^{1}$, $p_{p}^{2}$, $r_{p}^{1}$, $r_{p}^{2}$, $V_{p}$, $U_{p}'$, $q_{p}^{1}$, $q_{p}^{2}$, $s_{p}^{1}$, $s_{p}^{2}$, $V_{p}'$\\
    (Core and shadow goods)\\
    \vspace{0.1in}
    $S_{1}$, $S_{2}$, $C_{1}^{L}$, $C_{1}^{R}$, $\cdots$, $ C_{p}^{L}$, $C_{p}^{R}$, $D_{1}^{1}$, $D_{1}^{2}$, $D_{1}^{3}$, $\cdots$, $D_{p}^{1}$, $D_{p}^{2}$, $D_{p}^{3}$\\
    (Separator and auxiliary goods followed by the dummy goods)
    %
    \caption{The instance used in proof of part (b) of \Cref{thm:EF1_EQ1_PO}. The path graph is constructed such that the goods in the top row are to the left of those in the bottom row.}
    \label{fig:EF1_PO_appendix}
\end{figure*}

\textbf{Agents:} As before, we have the \emph{main agents} of the \emph{positive} and \emph{negative} type for every main variable $x_i$, denoted by $a_{x_i}$ and $a_{\bar{x}_i}$, respectively. 
For every $i \in [p]$, the agent $a_{x_i}$ approves the goods $U_i^{}, V_i^{}, D_i^1,D_i^2,D_i^3$, while the agent $a_{\bar{x_i}}$ approves the goods $U_i^\pr, V_i^\pr, D_i^1,D_i^2,D_i^3$. We also introduce a \emph{shadow agent} for every shadow variable. If $p_{i}$ is a shadow variable occurring in core clause $U_{i}$ and auxiliary clause $C_{j}$, then the corresponding shadow agent $p_{i}$ approves the shadow goods $p_{i}^{1}, p_{i}^{2}$ and the auxiliary goods $C_{j}^{L}, C_{j}^{R}$. The valuations of the other shadow agents $r_{i}, q_{i}, s_{i}$ are defined analogously. Lastly, we introduce a \emph{separator agent} $a_0$ that approves the two separator goods $S_1, S_2$. 
This completes the construction of the reduced instance. Notice that the constructed instance is $(5,4)$\emph{-sparse}. Before presenting the proof of equivalence, we will prove a structural result in \Cref{lem:EF1+PO-utility}.

\begin{lemma}
\label{lem:EF1+PO-utility}
In any $\EF{1}+\PO{}$ allocation, the utility of the separator agent $a_{0}$ is equal to $2$. Moreover, for every $i \in [p]$, exactly one of $a_{x_{i}}$ or $a_{\bar{x_{i}}}$ is allocated the triplet of goods $\{D_{i}^{1}, D_{i}^{2}, D_{i}^{3}\}$.
\end{lemma}
\begin{proof} (of \Cref{lem:EF1+PO-utility})
Observe that in any \EF{1} and Pareto optimal allocation $A$, the separator goods $S_{1}, S_{2}$ must be allocated to separator agent $a_{0}$. Indeed, $S_{1}, S_{2}$ are valued only by $a_{0}$. If $S_{1}, S_{2}$ are allocated to two distinct agents in some allocation $A^\pr{}$, then $A^\pr{}$ can be shown to be Pareto dominated by another allocation identical to $A^\pr{}$ except for the assignment of separator goods to the separator agent. Otherwise, if $S_{1}, S_{2}$ are allocated to the same agent (different from $a_{0}$) in $A^\pr{}$, then \EF{1} is violated from $a_{0}$'s perspective. Therefore, the utility of the separator agent $a_{0}$ under any \EF{1} and Pareto optimal allocation is equal to $2$. This implies that no main or shadow agent can obtain utility from goods in both rows of \Cref{fig:EF1_PO_appendix}.

To prove the second part of the lemma, we first observe that for every $i \in [p]$, the goods $D_{i}^{1}, D_{i}^{2}, D_{i}^{3}$ must be assigned between the main agents $a_{x_{i}}$ and $a_{\bar{x_{i}}}$ in any Pareto optimal allocation. This is because these goods are approved only by $a_{x_{i}}$ and $a_{\bar{x_{i}}}$ and no other agent. Furthermore, these agents can obtain a utility of at most $2$ from the core goods. Therefore, any allocation $A$ in which one or more of the dummy goods $D_{i}^{1}, D_{i}^{2}, D_{i}^{3}$ are assigned to agents other than $a_{x_{i}}$ and $a_{\bar{x_{i}}}$ can be shown to be Pareto dominated by another allocation, say $A'$, that is identical to $A$ except for the assignment of these dummy goods, which are allocated exclusively among $a_{x_{i}}$ and $a_{\bar{x_{i}}}$. 

Next, suppose that both $a_{x_{i}}$ and $a_{\bar{x_{i}}}$ are allocated only the dummy goods $D_{i}^{1}, D_{i}^{2}, D_{i}^{3}$ in a Pareto optimal allocation, say $A$. Assume, without loss of generality, that the utilities of $a_{x_{i}}$ and $a_{\bar{x_{i}}}$ in $A$ are $1$ and $2$, respectively. Then, $A$ can be shown to be Pareto dominated by another allocation that is identical to $A$ with the exception that one of the core goods, say $U_i$, is assigned to $a_{x_{i}}$, and the triplet $\{D_{i}^{1}, D_{i}^{2}, D_{i}^{3}\}$ to $a_{\bar{x_{i}}}$, contradicting the Pareto optimality of $A$. Thus, the triplet of dummy goods  $\{D_{i}^{1}, D_{i}^{2}, D_{i}^{3}\}$ must be completely assigned to either $a_{x_{i}}$ or $a_{\bar{x_{i}}}$.
\end{proof}

\textbf{The Forward Direction.} Given a satisfying assignment $\tau$ for \LES{}, we will construct the desired allocation as follows:
\begin{itemize}
    \item Allocate the separator goods $S_{1}, S_{2}$ to the separator agent $a_{0}$.
    \item If $\tau(x_{i})=1$, then allocate $\{U_{i}, p_{i}^{1},p_{i}^{2}, r_{i}^{1}, r_{i}^{2}, V_{i}\}$ to agent $a_{x_{i}}$ and $\{D_{i}^{1}, D_{i}^{2}, D_{i}^{3}\}$ to agent $a_{\bar{x_{i}}}$. In addition, allocate $\{U_{i}', q_{i}^{1},q_{i}^{2}\}$ to $q_{i}$, and $\{s_{i}^{1}, s_{i}^{2}, V_{i}'\}$ to $s_{i}$. Recall that $q_{i}$ and $s_{i}$ are the shadow variables that appear as negated literals in the core clauses $U_{i}'$ and $V_{i}'$, respectively, along with $\bar{x_{i}}$.
    
    Otherwise, if $\tau(x_{i}) = 0$, then allocate $\{U'_{i}, q_{i}^{1}, q_{i}^{2}, s_{i}^{1}, s_{i}^{2}, V'_{i}\}$ to agent $a_{\bar{x}_{i}}$ and $\{D_{i}^{1}, D_{i}^{2}, D_{i}^{3}\}$ to agent $a_{x_{i}}$. In addition, allocate $\{U_{i}, p_{i}^{1},p_{i}^{2}\}$ to $p_{i}$, and $\{r_{i}^{1}, r_{i}^{2}, V_{i}\}$ to $r_{i}$.
    
    \item Finally, for every $j \in [p]$, allocate $\{C_{j}^{L}\}$ and $\{C_{j}^{R}\}$ to the two shadow agents whose corresponding literals satisfy the auxiliary clause $C_j$.
\end{itemize}

Notice that in the constructed allocation, each good is allocated to exactly one agent, and each agent receives a connected interval. Also, the utility of the separator agent is $2$, and exactly one agent corresponding to each variable receives a triplet of the corresponding dummy goods. Note that the utility of each main agent is either $2$ or $3$, and the utility of each shadow agent is either $1$ or $2$. Furthermore, any main agent is allocated at most two goods valued by any shadow agent. Hence, the constructed allocation is \EF{1}.

We will now argue that the above allocation, say $A$, is Pareto optimal. Suppose for contradiction, that another allocation $A'$ Pareto dominates $A$. The second part of \cref{lem:EF1+PO-utility} implies that in any \EF{1} and Pareto optimal allocation, for every $i \in [p]$, the main agents $a_{x_{i}}$ and $a_{\bar{x_{i}}}$ cannot both have utility $3$. Thus, the utilities of the main agents under $A$ and $A'$ should be equal. Furthermore, one of the main agents corresponding to each variable will be allocated shadow goods corresponding to a pair of shadow agents (either \{$p_{i}, r_{i}\}$ or $\{q_{i}, s_{i}\}$). This implies that for at least one of these pairs, the two shadow agents should each receive a utility of $1$ under $A'$. Hence, by a similar argument as above, all shadow agents will also have the same utility under $A$ and $A'$, establishing that $A'$ cannot Pareto dominate $A$, as desired.

\textbf{The Reverse Direction.} We will now show a way to recover an \LES{} assignment given a connected \EF{1} and Pareto optimal allocation, say $A$.

Since $A$ is \EF{1} and Pareto optimal, we know from \cref{lem:EF1+PO-utility} that the separator agent receives the two approved goods, and for each variable $x_{i}$, exactly one of the corresponding main agents $a_{x_{i}}$ or $a_{\bar{x_{i}}}$ receives the triplet of dummy goods $\{D_{i}^{1}, D_{i}^{2}, D_{i}^{3}\}$. By \EF{1}, the other main agent will achieve a utility $2$ by either receiving the interval $\{U_{i}, p_{i}^{1}, p_{i}^{2}, r_{i}^{1}, r_{i}^{2}, V_{i}\}$ or $\{U_{i}', q_{i}^{1}, q_{i}^{2}, s_{i}^{1}, s_{i}^{2}, V_{i}'\}$. This, in turn, forces \emph{at least} one pair of shadow agents---either $\{p_{i}, r_{i}\}$ or $\{q_{i}, s_{i}\}$---to obtain their utilities from the auxiliary goods. Note that for any such pair, both agents will have a utility of at least $1$ due to \EF{1} condition.

We will now show that \emph{exactly} one of the two pairs of shadow agents derive their utility from the shadow goods, while the other pair meets the utility requirement though the auxiliary goods. Indeed, since there are $2p$ auxiliary goods (corresponding to $p$ auxiliary clauses), at most $2p$ shadow agents can obtain the desired utility from the auxiliary goods. Therefore, for every $i \in [p]$, exactly one pair of shadow agents---either $\{p_{i}, r_{i}\}$ or $\{q_{i}, s_{i}\}$---are assigned shadow goods, while the other pair receives auxiliary goods.

Overall, we have that one set of $p$ main agents gets exactly two core goods each (we will refer them as the ``lucky'' agents), while the other set of $p$ main agents gets three dummy goods each (the ``unlucky'' agents). Notice that the two main agents corresponding to a main variable cannot both be lucky, nor can both be unlucky due to the argument presented in \Cref{lem:EF1+PO-utility}.

This brings us to a natural way of deriving an \LES{} assignment $\tau$ from the allocation $A$. If the main agent of the positive (respectively, negative) type is unlucky, then we let $\tau(x_i) = 0$ (respectively, $\tau(x_i) = 1$). Furthermore, if $A$ allocates a core good to a shadow agent, then the corresponding shadow variable is set to $0$, while shadow variables corresponding to shadow agents who receive auxiliary goods are set to $1$. Note that exactly $2p$ of the $4p$ shadow variables are set to $1$ under this assignment and there are no conflicting assignments, implying that $\tau$ is indeed a valid solution to the \LES{} instance. This completes the proof of part (b) of \Cref{thm:EF1_EQ1_PO}.
\end{proof}

\begin{proof}(of part (d))
To prove part $(d)$, we adapt the construction in part (b) with a small change: For every $i \in [p]$, we introduce four auxiliary goods $C_{i}^{L_{1}}, C_{i}^{L_{2}}, C_{i}^{R_{1}}, C_{i}^{R_{2}}$ instead of the original two $C_{i}^{L}, C_{i}^{R}$. Note that such an instance is $(6,4)$\emph{-sparse}. We adapt the changes in the construction to the allocation constructed in the forward direction by replacing $C_{i}^{L}$ (respectively, $C_{i}^{R}$) with the set of goods $\{C_{i}^{L_{1}}, C_{i}^{L_{2}}\}$ (respectively, $\{C_{i}^{R_{1}}, C_{i}^{R_{2}}\}$). In the reverse direction, it is possible that under the given allocation, say $A$, the \emph{main} agents no longer receive all three dummy goods. Similar to the argument in part (c), we can construct an allocation $A'$ that is identical to $A$ except that we allocate the triplet of dummy goods $\{D_{i}^{1}, D_{i}^{2}, D_{i}^{3}\}$ to the corresponding main agent. At this stage, with a similar argument as in the reverse direction of part (b), we can recover a satisfying \LES{} assignment.
\end{proof}

\section{Proof of Theorem~\ref{thm:EQ1+POstar-polytime}}
\label{sec:Proof_EQ1+POstar-polytime}

Recall the statement of \Cref{thm:EQ1+POstar-polytime}.
\EQonePOstarcomp*

We will start by describing the algorithm underlying this result, which, in turn, builds on Algorithm~\ref{alg:EQ1+Complete}. This will be followed by a formal proof of \Cref{thm:EQ1+POstar-polytime}.

\textbf{Description of the algorithm for \EQ{1} and \PO{}* allocations:} Let $\sigma := (a_{1}, a_{2}, \ldots, a_{n})$. Our algorithm for \Cref{thm:EQ1+POstar-polytime} consists of four phases.
Phases 1 and Phase 2 are identical to those in Algorithm~\ref{alg:EQ1+Complete}, and are used to find the optimal egalitarian welfare $\theta$ and the leftmost \emph{$\theta$-unsafe} agent $a_{i}$, respectively. Recall that in Phase 2, we also fix the allocations of the agents $a_{1}, a_{2}, \ldots, a_{i-1}$.

In the third phase, which we denote by Phase 3*, we partition the agents $a_{i}, \ldots, a_{n}$ in two \emph{groups} as follows: We start with a partial allocation of the first $i-1$ agents $a_{1}, a_{2}, \ldots, a_{i-1}$, and then consider the remaining agents sequentially from left to right. That is, in round $j \in \{i, i+1, \dots, n\}$, we consider the leftmost unallocated agent according to $\sigma$, namely $a_{j}$. 
Starting with the leftmost available good, we allocate a minimal bundle worth $\theta+1$ to $a_{j}$ (note that $\theta+1$ is a realizable utility value under binary valuations). Next, the algorithm checks whether there exists a connected and $\sigma$-consistent allocation such that each subsequent agent receives utility $\theta$ (this step is similar to that in Algorithm~\ref{alg:EQ1+Complete}). 
If the check passes (i.e., if there is a feasible partial allocation where the agents $a_{j+1}, \ldots, a_{n}$ receive utility $\theta$ each), then we assign $a_{j}$ to \emph{group 1} and allocate to it the minimal bundle with utility $\theta+1$ (this is a temporary allocation).
Otherwise, we assign $a_{j}$ to \emph{group 2} and allocate to it the minimal bundle worth $\theta$. The above procedure is repeated for all subsequent agents, following which the algorithm proceeds to the fourth phase.

In Phase 4*, we finalize the allocation of the agents $a_{i}, \ldots, a_{n}$ (recall that the allocation in Phase 3* is only tentative). At first, we mimic the allocation for agents $a_{i}, \ldots, a_{n-1}$ from Phase 3*, and allocate the remaining goods to $a_{n}$.
In this allocation, if $u_{n}(A_{n}) \leq \theta+1$, then the algorithm finalizes the bundles of \emph{all} agents and returns the allocation. Otherwise, the algorithm performs a \emph{right-to-left} scan of the path $ G $ similar to Phase 3 of Algorithm~\ref{alg:EQ1+Complete}. 
In particular, starting from the rightmost available good, the algorithm moves leftwards along $G$ and iteratively assigns minimal connected bundles with utility $\theta+1$ to group 1 or $\theta$ to group 2 agents 
in the reverse order $a_{n}, a_{n-1}, \ldots, a_{i+1}$. The remaining goods are assigned to agent $a_i$ and the final allocation is returned as the output.

\begin{proof} (of \Cref{thm:EQ1+POstar-polytime})
First, observe that the set of goods allocated in Phase 4* at least contains \emph{all} the goods allocated in a temporary allocation of Phase 3* (it may contain some additional goods). 
This is because, during a \emph{left-to-right} partial temporary allocation in Phase 3*, we may not consider leftover goods to the right of the allocated bundle for agent $a_{n}$.
Hence, in the final allocation in Phase 4*, algorithm has enough goods such that each group 1 agent receives a utility of at least $\theta+1$, and each group 2 agent receives a utility of at least $\theta$.

Next, we show that the allocation $A$ returned by the algorithm is $\sigma$-consistent, complete, and \EQ{1}. 
Notice that for the two cases in Phase 4*, in the last iteration, we allocate the leftover goods to agent $a_{n}$ or agent $a_{i}$; hence, completeness follows trivially. 
Also, $A$ is $\sigma$-consistent because the algorithm maintains this property at every step. 
Note that each of the agents $a_{1}, \ldots, a_{i-1}$ receives a bundle with utility $\theta+1$ under allocation $A$. 
In Phase 4*, consider the case when final allocation is a completion of temporary partial allocation from Phase 3* by assigning leftover goods to agent $a_{n}$.
Here, it is easy to see that the agents $a_{i}, \ldots, a_{n}$ are each allocated bundles with utility $\theta$ or $\theta+1$. Hence, the allocation is \EQ{1}.
For the other case, when the final allocation is built with a right-to-left traversal and the leftover goods are assigned to agent $a_{i}$, it is easy to see that, except for agent $a_{i}$ all other agents receive a bundle with utility either $\theta$ or $\theta+1$. Moreover, $a_{i}$ belongs to group 2 (using the definition of $\theta$-unsafe agent), and it receives a bundle with utility at least $\theta$.
Now, let $S, S'$ be the set of goods allocated to agent $a_{i}$ under the allocation $A$, and allocation (say $A'$) by Algorithm~\ref{alg:EQ1+Complete} for when we run it on the same instance respectively. 
Observe that $S \subseteq S'$ since the agents $a_{i+1}, \ldots, a_{n}$ receive a minimal bundle with utility either $\theta$ or $\theta+1$ under the allocation $A'$ while these agents receive a minimal bundle of utility exactly $\theta$ under the allocation $A$. 
At this stage, just like in the proof of Algorithm~\ref{alg:EQ1+Complete}, we can conclude that $u_{i}(S) = \theta$. 
Hence, the allocation $A$ is indeed \EQ{1}.

Finally, we show that $A$ is \PO{}* with a proof by contradiction. 
Let $B$ be a $\sigma$-consistent \EQ{1} allocation that Pareto dominates $A$.
From \Cref{thm:EQ1+complete-polytime}, we know that the optimal egalitarian welfare for any connected and $\sigma$-consistent allocation is $\theta$. 
Hence, each agent receives a bundle with utility either $\theta$ or $\theta+1$ under allocation $A$ due to the way our algorithm works.
Moreover, each agent receives a bundle with utility either $\theta$ or $\theta+1$ under allocation $B$ as $\theta$ is the optimal egalitarian welfare for the instance and allocation $B$ is \EQ{1}.
Let $j$ be the leftmost agent such that $u_{j}(A_{j}) < u_{j}(B_{j})$. 
We claim that $j > i$ where $a_{i}$ is the $\theta$-unsafe agent.
This is because the agents $a_{1}, a_{2}, \ldots, a_{i-1}$ each receive a bundle with utility $\theta+1$ under allocation $A$ and $a_{i}$ is the $\theta$-unsafe agent. 
It is easy to see that $u_{j}(B_{j})= \theta+1$. 
Let $A'$ be an allocation which is identical to allocation $A$ for all agents $a_{1}, a_{2}, \ldots, a_{j-1}$, and allocates a minimal connected bundle to $a_{j}$ such that $u_{j}(A_{j}') = \theta+1$.
Furthermore, let $A^{*}$ (respectively, $B^{*}$) be the set of goods to the right of bundle $A_{j}'$ (respectively, $B_{j}$).
We claim that $B^{*} \subseteq A^{*}$. This is because for all $\ell < j, u_{\ell}(A_{\ell}) = u_{\ell}(B_{\ell})$, and our algorithm allocated \emph{minimal} bundles.
But since $u_{j}(A_{j}) = \theta$, in Phase 3*, our algorithm labeled $a_{j}$ as group 1 agent. 
This implies that the set of goods $A^{*}$ is not sufficient to ensure a utility $\theta$ for all subsequent agents. 
Since $B^{*} \subseteq A^{*}$, the allocation $B$ is not \EQ{1} which is a contradiction.

Finally, let us turn to the running time analysis. 
Since we only consider \emph{binary} valuations, the list $L$ of all distinct realizable utility values contains at most $m$ distinct values, and can be precomputed in $\O(nm)$ time. 
By a similar running time analysis as in the proof of \Cref{thm:EQ1+complete-polytime}, it follows that the total running time for Phase 1 is $\OO(m^{2})$, and that for Phase 2 is $\OO(nm)$.

In Phase 3*, for each fixed $j$, in order to decide the group of agent $a_{j}$, each of the unallocated goods is considered towards at most one bundle. Hence, the total running time for this and each subsequent iteration is $\OO(m)$. Since there are at most $n$ iterations, the total running time for Phase 3* is $\OO(nm)$. In Phase 4*, we finalize the allocation of the agents $a_{i}, \ldots, a_{n}$ by constructing at most two complete allocations corresponding to the two cases. Each good is considered towards at most one bundle in each of these allocations. Thus, the algorithm requires $\OO(m)$ time in Phase 4*. Hence, the overall running time of our algorithm is $\OO(m^{2}+nm)$.
\end{proof}

We close this section by noting that the problem of computing \PO{}* allocations remains an interesting open question for general monotone valuations. Our algorithm for this problem does not extend too far beyond the binary regime, as the following example shows: Consider an instance with four goods $v_1, v_2, v_3, v_4$ and two agents $a_1, a_2$ with valuations $u_1=(1,3,1,0)$, $u_2=(0,0,0,2)$. Suppose the agent ordering is $\sigma = (1,2)$. 
The optimal egalitarian welfare is $\theta = 2$, and $a_{2}$ is the leftmost \emph{$\theta$-unsafe} agent. On this instance, our algorithm returns the \EQ{1} allocation $(\{v_{1}, v_{2}\}, \{v_{3}, v_{4}\})$, which is Pareto dominated by the \EQ{1} allocation $(\{v_{1}, v_{2}, v_{3}\}, \{v_{4}\})$.

\end{appendices}

\end{document}